\documentclass[a4paper, 11pt]{scrartcl}

\pdfoutput=1

\usepackage{microtype}
\usepackage[
    pdfusetitle,
    colorlinks,
    linkcolor={Blue},
    citecolor={black},
    urlcolor={black}
]{hyperref}
\usepackage[utf8]{inputenc}
\usepackage[T1]{fontenc}
\usepackage{lmodern}
\usepackage{subcaption}
\usepackage[USenglish]{babel}
\usepackage[dvipsnames]{xcolor}

\usepackage{adjustbox}
\usepackage{amsmath}
\usepackage{amssymb}
\usepackage{amsthm}
\usepackage{authblk}

\usepackage{newtxtext}
\usepackage{newtxmath}

\theoremstyle{plain}
\newtheorem{theorem}{Theorem}[section]
\newtheorem*{theorem*}{Theorem}
\newtheorem{proposition}[theorem]{Proposition}
\newtheorem*{proposition*}{Proposition}
\newtheorem{lemma}[theorem]{Lemma}
\newtheorem*{lemma*}{Lemma}
\newtheorem{corollary}[theorem]{Corollary}
\newtheorem*{corollary*}{Corollary}
\newtheorem{claim}[theorem]{Claim}
\newtheorem*{claim*}{Claim}

\theoremstyle{definition}

\newtheorem*{definition*}{Definition}

\theoremstyle{remark}
\newtheorem{observation}[theorem]{Observation}
\newtheorem*{observation*}{Observation}
\newtheorem{example}[theorem]{Example}
\newtheorem*{example*}{Example}

\newenvironment{claimproof}{\begin{proof}}{\end{proof}}

\title{Using Color Refinement to Boost Enumeration and Counting for 
  Acyclic~CQs of Binary Schemas}

\author[1]{Cristian Riveros}
\author[2]{Benjamin Scheidt}
\author[2]{Nicole Schweikardt}
\affil[1]{Pontificia Universidad Católica de Chile}
\affil[2]{Humboldt-Universität zu Berlin, Germany}

\date{}

\usepackage{stmaryrd}
\usepackage{xspace}
\usepackage{amsmath}
\usepackage{paralist}
\usepackage{varwidth}

\usepackage{tikz}
\usetikzlibrary{automata, graphs,positioning,chains,arrows,decorations.pathmorphing}

\newenvironment{me}{\begin{enumerate}[(1)]}{\end{enumerate}}
\newenvironment{mea}{\begin{enumerate}[(a)]}{\end{enumerate}}

\renewcommand{\epsilon}{\varepsilon}

\newtheorem{fact}[theorem]{Fact}

\newenvironment{exampleWithEndmarker}{\begin{example}}{\mbox{}\hfill$\lrcorner$\end{example}}      
 
\newcommand{\exrel}[1]{\textbf{#1}}
\newcommand{\exdata}[1]{\texttt{#1}}
 
 \tikzset{
 	defaultstyle/.style={
 		>=stealth, 
 		semithick, 
 		auto,
 		initial text= {},
 		initial distance= {3mm},
 		accepting distance= {3mm}}}
 	
\newcommand{\exsigma}{\sigma^{\operatorname{ex}}}
\newcommand{\exdb}{D^{\operatorname{ex}}}
\newcommand{\exgraph}{\GOne^{\operatorname{ex}}}
\newcommand{\exdbcol}{\exdb_{\operatorname{col}}}
   
\newcommand{\nc}[1]{\newcommand{#1}}
\newcommand{\rnc}[1]{\renewcommand{#1}}

\nc{\myparagraph}[1]{\paragraph{#1}}
\rnc{\leq}{\ensuremath{\leqslant}}
\rnc{\geq}{\ensuremath{\geqslant}}

\rnc{\le}{\leq}
\rnc{\ge}{\geq}

\nc{\isdef}{\ensuremath{:=}}
\nc{\deff}{\isdef}

\nc{\set}[1]{\ensuremath{\{#1\}}}
\nc{\setsize}[1]{\ensuremath{|#1|}}
\nc{\Setsize}[1]{\ensuremath{\big|#1\big|}}
\nc{\Set}[1]{\ensuremath{\big\{#1\big\}}}
\nc{\setc}[2]{\set{#1 \, : \, #2}}
\nc{\Setc}[2]{\Set{#1 \, : \, #2}}

\nc{\aufgerundet}[1]{\ensuremath{\lceil #1 \rceil}}
\nc{\abgerundet}[1]{\ensuremath{\lfloor #1 \rfloor}}

\nc{\dcup}{\ensuremath{\dot\cup}}

\nc{\ov}[1]{\ensuremath{\overline{#1}}}

\nc{\NN}{\ensuremath{\mathbb{N}}}
\nc{\NNpos}{\ensuremath{\NN_{\scriptscriptstyle\geq 1}}}
\nc{\RR}{\ensuremath{\mathbb{R}}}
\nc{\RRpos}{\ensuremath{\RR_{\scriptscriptstyle\geq 0}}}
\nc{\QQ}{\ensuremath{\mathbb{Q}}}
\nc{\QQpos}{\ensuremath{\QQ_{\scriptscriptstyle\geq 0}}}

\nc{\und}{\ensuremath{\wedge}}
\nc{\Und}{\ensuremath{\bigwedge}}
\nc{\oder}{\ensuremath{\vee}}
\nc{\Oder}{\ensuremath{\bigvee}}
\nc{\nicht}{\ensuremath{\neg}}
\nc{\impl}{\ensuremath{\to}}
\nc{\gdw}{\ensuremath{\leftrightarrow}}

\nc{\Semijoin}{\ensuremath{\ltimes}}

\nc{\free}{\ensuremath{\textrm{\upshape free}}}
\nc{\quant}{\ensuremath{\textrm{\upshape quant}}}
\nc{\ar}{\ensuremath{\operatorname{ar}}}

\nc{\Structure}[1]{\ensuremath{\mathcal{#1}}}
\nc{\A}{\Structure{A}}
\nc{\B}{\Structure{B}}
\nc{\C}{\Structure{C}}

\nc{\isom}{\ensuremath{\cong}}

\nc{\querycont}{\ensuremath{\sqsubseteq}}
\nc{\eval}[2]{\ensuremath{#1(#2)}}
\nc{\semantik}[1]{\ensuremath{\left\llbracket#1\right\rrbracket}}
\nc{\sem}[1]{{\semantik{#1}}}

\nc{\CanDB}[1]{\ensuremath{\A_{#1}}} %
\nc{\CanTup}[1]{\ensuremath{t_{#1}}} %

\newcommand{\queryphi}{\varphi}

\newcommand{\relS}{S} %
 
\newcommand{\relT}{T} %

\newcommand{\relE}{E} %

\nc{\Vars}{\ensuremath{\textrm{\upshape vars}}}
\nc{\vars}{\Vars}
\nc{\Cons}{\ensuremath{\textrm{\upshape cons}}}
\nc{\cons}{\Cons}
\nc{\atoms}{\ensuremath{\textrm{\upshape atoms}}}
\nc{\Atoms}{\atoms}
\nc{\Adom}{\ensuremath{\textrm{\upshape adom}}}
\nc{\adom}[1]{\ensuremath{\Adom(#1)}}
\nc{\dom}[1]{\ensuremath{\textrm{\upshape dom}(#1)}}

\newcommand{\DBone}[1]{}

\newcommand{\parent}{\pointerfont{parent}}

\nc{\arrayfont}[1]{\ensuremath{\texttt{#1}}}

\nc{\card}[1]{\ensuremath{|#1|}}

\newcommand{\assign}{\ensuremath{\alpha}}

\nc{\insertp}{\textsc{Insert}}
\nc{\cleanup}{\textsc{cleanUp}}
\nc{\cleanups}{\textsc{cleanUp'}}

\nc{\Yes}{\texttt{\upshape true}}
\nc{\No}{\texttt{\upshape false}}

\nc{\Dom}{\ensuremath{\textbf{dom}}}
\nc{\Var}{\ensuremath{\textbf{var}}}
\nc{\schema}{\ensuremath{\sigma}}
\nc{\DB}{\ensuremath{D}} %
\nc{\DBstrich}{\ensuremath{D'}}
\nc{\DBstart}{\ensuremath{{\DB_0}}}
\nc{\DBempty}{\ensuremath{{\DB_{\emptyset}}}}

\nc{\DS}{\ensuremath{\mathtt{D}}}

\rnc{\phi}{\queryphi}

\nc{\UpdateFont}[1]{\ensuremath{\textsf{#1}}}
\nc{\Delete}{\UpdateFont{delete}}
\nc{\Insert}{\UpdateFont{insert}}
\nc{\Update}{\UpdateFont{update}}

\nc{\AlgoFont}[1]{\ensuremath{\textbf{#1}}}
\nc{\PREPROCESS}{\AlgoFont{preprocess}}
\nc{\INIT}{\AlgoFont{init}}
\nc{\UPDATE}{\AlgoFont{update}}
\nc{\ENUMERATE}{\AlgoFont{enumerate}}
\nc{\COUNT}{\AlgoFont{count}}
\nc{\ANSWER}{\AlgoFont{answer}}
\nc{\TEST}{\AlgoFont{test}}

\nc{\EOE}{\texttt{\upshape EOE}\xspace}

\nc{\preprocessingtime}{\ensuremath{t_{\operatorname{prep}}}}
\nc{\delaytime}{\ensuremath{t_{\operatorname{delay}}}}
\nc{\preprocessingtimefunc}{\ensuremath{f_{\operatorname{prep}}}}
\nc{\delaytimefunc}{\ensuremath{f_{\operatorname{delay}}}}
\nc{\indexingtime}{\ensuremath{t_{\operatorname{index}}}}
\nc{\indexingtimefunc}{\ensuremath{f_{\operatorname{index}}}}
\nc{\booltime}{\ensuremath{t_{\booltask}}}
\nc{\countingtime}{\ensuremath{t_{\counttask}}}
\nc{\booltimefunc}{\ensuremath{f_{\booltask}}}
\nc{\countingtimefunc}{\ensuremath{f_{\counttask}}}

\nc{\preprocessingtimehat}{\ensuremath{\hat{t}_p}}
\nc{\inittimehat}{\ensuremath{\hat{t}_i}}
\nc{\delaytimehat}{\ensuremath{\hat{t}_d}}
\nc{\updatetimehat}{\ensuremath{\hat{t}_u}}
\nc{\answertimehat}{\ensuremath{\hat{t}_a}}
\nc{\countingtimehat}{\ensuremath{\hat{t}_c}}
\nc{\testingtimehat}{\ensuremath{\hat{t}_t}}

\nc{\phiBTypical}{\ensuremath{\phi'_{\relS\text{-}\relE\text{-}\relT}}}
\nc{\phiJTypical}{\ensuremath{\phi_{\relS\text{-}\relE\text{-}\relT}}}
\nc{\phiET}{\ensuremath{\phi_{\relE\text{-}\relT}}}

\nc{\restrict}[2]{\ensuremath{{#1}_{|#2}}}
\nc{\extend}[3]{\ensuremath{{#1}\tfrac{#3}{#2}}}
\nc{\emptyassign}{\ensuremath{\emptyset}}
\nc{\Assign}[2]{\ensuremath{\frac{#2}{#1}}}

\nc{\vroot}{\ensuremath{\varv_{\textsl{root}}}}

\nc{\pointerfont}[1]{\textit{#1}}

\nc{\varitem}[1]{\ensuremath{v^{#1}}}
\nc{\assitem}[1]{\ensuremath{\assign^{#1}}}
\nc{\constitem}[1]{\ensuremath{a^{#1}}}
\nc{\parentitem}[1]{\ensuremath{\parent^{#1}}}
\nc{\childitem}[2]{\ensuremath{\pointerfont{child}^{#1}_{#2}}}
\nc{\llist}[2]{\ensuremath{\mathcal{L}_{#2}^{#1}}}
\nc{\startlist}{\ensuremath{\mathcal{L}_{\text{\upshape start}}}\xspace}
\nc{\nextlistitem}[1]{\ensuremath{\pointerfont{next-listitem}^{#1}}}
\nc{\prevlistitem}[1]{\ensuremath{\pointerfont{prev-listitem}^{#1}}}
\nc{\countitem}[1]{\ensuremath{C_{\textit{below}}^{#1}}}
\nc{\desc}[1]{\ensuremath{\text{desc}}}

\nc{\Null}{\ensuremath{0}}

\nc{\arrayA}{\arrayfont{A}}
\nc{\arrayB}{\arrayfont{B}}
\nc{\arrayC}{\arrayfont{C}}
\nc{\arrayE}{\arrayfont{E}}

\nc{\ITEMS}{\mathcal{I}}

\nc{\NIL}{\textsc{nil}}

\nc{\TupleSet}{\ensuremath{\mathcal{T}}}
\nc{\ResultSet}{\ensuremath{\mathcal{R}}}
\nc{\SkipArrayNext}[1]{\ensuremath{\mathsf{skip}[#1].\mathsf{next}}}
\nc{\SkipArrayPrev}[1]{\ensuremath{\mathsf{skip}[#1].\mathsf{prev}}}
\nc{\AlgoA}{\ensuremath{\mathbb{A}}}
\nc{\nil}{\texttt{nil}\xspace}
\nc{\SkipStart}{\ensuremath{\mathsf{sk{-}start}}}
\nc{\tup}{\ensuremath{\ov{t}}}
\nc{\tups}{\ensuremath{\ov{s}}}
\nc{\prozvisit}{\ensuremath{\textsc{Visit}}}
\nc{\prozvisitrev}{\ensuremath{\textsc{Visit}^{-1}}}
\nc{\tut}{\ensuremath{t}}
\nc{\enumprev}{\ensuremath{\vartriangleleft}}
\nc{\SkipLast}{\ensuremath{\mathsf{sk{-}last}}}

\nc{\lllist}{\ensuremath{\mathcal{L}}}
\nc{\pllist}{\ensuremath{\mathcal{L}^+}}
\nc{\milist}{\ensuremath{\mathcal{L}^-}}
\nc{\cilist}{\ensuremath{\mathcal{L}^\circ}}
\nc{\numitmpl}{\ensuremath{+{-}\text{on}{-}\text{path}}}
\nc{\numitmmi}{\ensuremath{-{-}\text{on}{-}\text{path}}}
\nc{\numitmci}{\ensuremath{\circ{-}\text{on}{-}\text{path}}}
\nc{\DBnew}{\ensuremath{\DB_{\text{new}}}}
\nc{\DBold}{\ensuremath{\DB_{\text{old}}}}
\nc{\liitmpl}{\ensuremath{\mathcal{L}^{+{-}\text{on}{-}\text{path}}}}
\nc{\liitmmi}{\ensuremath{\mathcal{L}^{-{-}\text{on}{-}\text{path}}}}
\nc{\liitmci}{\ensuremath{\mathcal{L}^{\circ{-}\text{on}{-}\text{path}}}}
\nc{\ITEMSres}[1]{\ensuremath{\ITEMS|_{#1}}}

\nc{\prVisit}{\textsc{Visit}}
\nc{\prVisitRes}{\textsc{VisitRes}}
\nc{\prEnumWithItem}{\textsc{EnumWithItem}}
\nc{\prFindItems}{\textsc{FindItems}}

\nc{\SUBW}{\ensuremath{\textit{subw}}}
\nc{\ADW}{\ensuremath{\textit{adw}}}
\nc{\fcSUBW}{\ensuremath{\textit{fc-subw}}}
\nc{\fcFHW}{\ensuremath{\textit{fc-fhw}}}
\nc{\fcGHW}{\ensuremath{\textit{fc-ghw}}}

\nc{\TD}{\ensuremath{\textit{TD}}}
\nc{\FDecom}{\ensuremath{\textit{F}}}

\nc{\emptytuple}{\ensuremath{()}}
\nc{\emptyword}{\ensuremath{\varepsilon}}

\nc{\proj}{\ensuremath{\pi}}
\nc{\select}{\ensuremath{\sigma}}

\nc{\FD}{\ensuremath{\delta_{\textit{fd}}}}
\nc{\IND}{\ensuremath{\delta_{\textit{ind}}}}
\nc{\INDtilde}{\ensuremath{\tilde{\delta}_{\textit{ind}}}}
\nc{\SD}{\ensuremath{\delta_{\textit{sd}}}}
\nc{\CC}{\ensuremath{\delta_{\textit{cc}}}}

\nc{\DEP}{\ensuremath{\delta}}
\nc{\CONSTR}{\ensuremath{\Sigma}}

\nc{\qSET}{\ensuremath{q_{\textit{S-E-T}}}}
\nc{\pSET}{\ensuremath{p_{\textit{S-E-T}}}}

\nc{\qET}{\ensuremath{q_{\textit{E-T}}}}

\nc{\Ans}{\ensuremath{\textit{Ans}}}

\nc{\query}{\ensuremath{Q}}
\nc{\qatom}{\ensuremath{\myatom}}

\nc{\HG}{\ensuremath{\mathcal{H}}}
\nc{\Nodes}{\ensuremath{V}}
\nc{\Edges}{\ensuremath{E}}

\nc{\HD}{\ensuremath{\textit{H}}}
\nc{\Tree}{\ensuremath{T}}
\nc{\rootedTree}{\ensuremath{\hat{\Tree}}}

\nc{\treenode}{\ensuremath{t}}
\nc{\treenodeparent}{\ensuremath{p}}
\nc{\parentnode}{\treenodeparent} 
\nc{\treeroot}{\ensuremath{r}}
\nc{\Bag}{\ensuremath{\textit{bag}}}
\nc{\Cover}{\ensuremath{\textit{cover}}}

\nc{\FHD}{\ensuremath{\textit{F}}}
\nc{\Weight}{\ensuremath{\textit{weight}}}

\nc{\Width}{\ensuremath{\textit{width}}}
\nc{\GHW}{\ensuremath{\textit{ghw}}}
\nc{\FHW}{\ensuremath{\textit{fhw}}}

\nc{\freetreenodes}{\ensuremath{U}}
\nc{\prunedTree}{\ensuremath{\tilde{\Tree}}}
\nc{\prunedSchema}{\ensuremath{\tilde{\schema}}}
\nc{\prunedDB}{\ensuremath{\tilde{\DB}}}
\nc{\prunedDBold}{\ensuremath{\prunedDB_{\textit{old}}}}
\nc{\prunedDBnew}{\ensuremath{\prunedDB_{\textit{new}}}}
\nc{\prunedQuery}{\ensuremath{\tilde{\query}}}

\nc{\Start}{\ensuremath{\texttt{fetch-first}}}
\nc{\Next}{\ensuremath{\texttt{fetch-next}}}
\nc{\TestTuple}{\ensuremath{\texttt{test}}}
\nc{\PositionCursor}{\ensuremath{\texttt{position-cursor}}}
\nc{\AccessJth}{\ensuremath{\texttt{access}}}
\nc{\RankTuple}{\ensuremath{\texttt{rank}}}

\nc{\Mapping}[1]{\ensuremath{\tilde{#1}}}
\nc{\MappingR}{\ensuremath{\tilde{R}}}

\nc{\Algo}[1]{\ensuremath{\textsc{#1}}}

\nc{\True}{\ensuremath{\texttt{true}}}
\nc{\False}{\ensuremath{\texttt{false}}}

\nc{\QueryClass}{\ensuremath{\mathcal{Q}}}
\nc{\fcACQ}{\ensuremath{\textsf{\upshape{fc-ACQ}}[\sigma]}}
\nc{\fcACQOne}{\ensuremath{\textsf{\upshape{fc-ACQ}}[\sigmaOne]}}
\nc{\fcACQci}{\ensuremath{\textsf{\upshape{fc-ACQ}}[\cisigma]}}

\nc{\AllDBs}[1]{\ensuremath{\textsf{DB}[#1]}}

\nc{\taskdescription}[1]{\textsl{\textsf{#1}}}
\nc{\booltask}{\taskdescription{bool}}
\nc{\counttask}{\taskdescription{count}}
\nc{\enumtask}{\taskdescription{enum}}

\nc{\IndexingProblem}[2]{\ensuremath{\textup{\textsc{IndexingEval}}(#1,#2)}}
\nc{\IndexingProblemGeneral}{\IndexingProblem{\sigma}{\QueryClass}}
\nc{\IndexingProblemOurs}{\IndexingProblem{\sigma}{\fcACQ}}

\nc{\valuation}{\ensuremath{\nu}}
\nc{\val}{\ensuremath{\valuation}}

\nc{\sigmaOne}{\ensuremath{\bar{\sigma}}}
\nc{\DOne}{\ensuremath{\bar{D}}}
\nc{\GOne}{\ensuremath{\bar{G}}}
\nc{\VOne}{\ensuremath{\bar{V}}}
\nc{\EOne}{\ensuremath{\bar{E}}}
\nc{\QOne}{\ensuremath{\bar{Q}}}
\nc{\vl}{\ensuremath{\textnormal{\textsf{vl}}}}
\nc{\el}{\ensuremath{\textnormal{\textsf{el}}}}
\nc{\VLabels}{\ensuremath{L_{\VOne}}}
\nc{\ELabels}{\ensuremath{L_{\EOne}}}
\nc{\MyPlus}{\ensuremath{\rightarrow}}
\nc{\MyMinus}{\ensuremath{\leftarrow}}

\nc{\img}{\ensuremath{\textrm{img}}}

\nc{\Neighbors}[5]{\ensuremath{{#1}_{#2}^{#3}({#4},{#5})}}
\nc{\NSucc}[3]{\Neighbors{N}{\rightarrow}{#1}{#2}{#3}}
\nc{\NPred}[3]{\Neighbors{N}{\leftarrow}{#1}{#2}{#3}}
\nc{\hatNSucc}[3]{\Neighbors{\widehat{N}}{\rightarrow}{#1}{#2}{#3}}
\nc{\hatNPred}[3]{\Neighbors{\widehat{N}}{\leftarrow}{#1}{#2}{#3}}

\nc{\Numbers}[5]{\ensuremath{{#1}_{#2}^{#3}({#4},{#5})}}
\nc{\numSucc}[3]{\Numbers{\#}{\rightarrow}{#1}{#2}{#3}}
\nc{\numPred}[3]{\Numbers{\#}{\leftarrow}{#1}{#2}{#3}}
\nc{\numNeigh}[3]{\Numbers{\#}{d}{#1}{#2}{#3}}
\nc{\hatnumSucc}[3]{\Numbers{\widehat{\#}}{\rightarrow}{#1}{#2}{#3}}
\nc{\hatnumPred}[3]{\Numbers{\widehat{\#}}{\leftarrow}{#1}{#2}{#3}}

\nc{\Numb}[1]{\ensuremath{n_{#1}}}

\nc{\col}{\ensuremath{\textnormal{\textsf{col}}}}
\nc{\colAlt}{\col'}

\nc{\Dual}[1]{\ensuremath{\widetilde{#1}}}

\nc{\DSD}{\ensuremath{\textsf{\upshape{DS}}_D}}

\nc{\fcr}{\ensuremath{f_{\operatorname{col}}}}

\nc{\cisigma}{\ensuremath{\sigma_{\operatorname{col}}}}
\nc{\ciD}{\ensuremath{{D_{\operatorname{col}}}}}
\nc{\ciDn}{\ensuremath{{D^n_{\operatorname{col}}}}}
\nc{\ciQ}{\ensuremath{{Q_{\operatorname{col}}}}}
\nc{\ciQStrich}{\ensuremath{{Q'_{\operatorname{col}}}}}

\nc{\elmt}[3]{{\ensuremath{e^{#1}_{({#2},{#3})}}}}

\nc{\myEnum}{\ensuremath{\textsc{Enum}}}

\nc{\Parent}{\ensuremath{\textit{p}}}
\nc{\Children}{\ensuremath{\textit{ch}}}

\nc{\DataStructure}[2]{\textsf{\upshape{DS}}_{D,Q}}

\nc{\myatom}{\ensuremath{a}}

\renewcommand*{\mid}{\, : \,}

\newcommand*{\fd}{f_{\downarrow}}
\newcommand*{\ParentEdge}{\textit{pe}}

\tikzset{nudge/.code args={#1}{%
  \pgfkeysalso{transform canvas={xshift=#1}}%
}} 

\begin{document}

\maketitle

\begin{abstract}
  We present an index structure, called the color-index,
  to boost the evaluation of acyclic conjunctive queries (ACQs) over binary schemas.
  The color-index is based on the color refinement algorithm,
  a widely used subroutine for graph isomorphism testing algorithms.
  Given a database $D$, we use a suitable version of the color refinement algorithm to produce
  a stable coloring of $D$, an assignment from the active domain of $D$ to a set of colors $C_D$.
  The main ingredient of the color-index is a particular database $\ciD$
  whose active domain is $C_D$ and whose size is at most $|D|$.
  Using the color-index, we can evaluate any free-connex ACQ $Q$ over $D$ with preprocessing time
  $O(|Q| \cdot |\ciD|)$ and constant delay enumeration.
  Furthermore, we can also count the number of results of $Q$ over $D$
  in time $O(|Q| \cdot |\ciD|)$.
  Given that $|\ciD|$ could be much smaller than $|D|$
  (even constant-size for some families of databases),
  the color-index is the first index structure for evaluating free-connex ACQs that allows efficient
  enumeration and counting with performance that may be strictly smaller than the database size.
\end{abstract}

\noindent\fcolorbox{black}{gray!30}{%
\begin{minipage}[t]{\dimexpr\textwidth-2\fboxsep-2\fboxrule}%
	\centering\Large%
	This paper is superseded by \href{https://arxiv.org/abs/2601.04757}{arXiv:2601.04757}~\cite{RSS26} of the same authors, handling arbitrary relational schemas --- not just binary ones.
\end{minipage}%
}

\section{Introduction}\label{sec:introduction}
An important part of database systems are \emph{index structures} that provide
efficient access to the stored data and help to accelerate query evaluation.
Such index structures typically rely on hash tables or balanced trees such as B-trees or B$^+$-trees,
which boost the performance of value queries~\cite{ramakrishnan2003database}.
Another recent example is indices for worst-case optimal join algorithms~\cite{ngo2018worst}.
For example, \emph{Leapfrog Triejoin}~\cite{DBLP:conf/icdt/Veldhuizen14},
a simple worst-case optimal algorithm for evaluating multiway-joins on relational databases,
is based on so-called trie iterators for boosting the access under different join orders.
These trie indices have recently improved in the use of time and space~\cite{ArroyueloHNRRS21}.
Typically, index structures are not constructed for supporting the evaluation of a single query,
but for supporting the evaluation of an entire class of queries.
This paper presents a novel kind of index structure called the \emph{color-index}
to boost the evaluation of free-connex acyclic conjunctive queries (fc-ACQs).

Acyclic conjunctive queries (ACQs) were introduced
in~\cite{DBLP:journals/jacm/BeeriFMY83,DBLP:journals/siamcomp/BernsteinG81}
and have since then received a lot of attention.
From Yannakakis' seminal paper~\cite{Yannakakis1981} it is known that
the result of every fixed ACQ $Q$ over a database $D$ can be computed in time
linear in the product of the database size $|D|$ and the output size $|\sem{Q}(D)|$.
For the particular subclass fc-ACQ of \emph{free-connex} ACQs,
it is even known to be linear in the sum $|D|+|\sem{Q}(D)|$:
the notion of \emph{free-connex} ACQs was introduced in the seminal paper by
Bagan, Durand, and Grandjean~\cite{Bagan.2007}, who improved Yannakakis' result as follows.
For any database $D$, during a preprocessing phase that takes time linear in $|D|$,
a data structure can be computed that allows to enumerate, without repetition,
all the tuples of the query result $\sem{Q}(D)$,
with only a \emph{constant} delay between outputting any two tuples.
Note that the data structure computed in the preprocessing phase is designed for the particular
query $Q$, and the above running time statement suppresses factors that depend on $Q$.
For evaluating a new query $Q'$, one has to start an entirely new preprocessing phase that,
again, takes time linear in $|D|$.

The main contribution of this paper is a novel index structure called the \emph{color-index}.
Upon input of a database $D$ of a binary schema
(i.e., the schema consists of relation symbols of arity at most 2),
the color-index can be built in time $O(|D|{\cdot}\log|D|)$ in the worst-case,
and for many $D$ the time is only $O(|D|)$.
The main ingredient of the color-index is the \emph{color database} $\ciD$.
By using the color-index, upon input of an arbitrary fc-ACQ $Q$ we can enumerate
the tuples of the query result $\sem{Q}(D)$ with constant delay after a preprocessing phase
that takes time linear in $|\ciD|$.
Compared to the above-mentioned result by Bagan, Durand, and Grandjean~\cite{Bagan.2007},
this accelerates the preprocessing time from $O(|D|)$ to $O(|\ciD|)$.
The size of $\ciD$ is $O(|D|)$ in the worst-case, but depending on the particular structure of $D$,
the database $\ciD$ may be substantially smaller than $D$; we show that even $O(1)$ is possible.
Apart from enumerating the query result with constant delay, the color-index
can also be used to count the number of tuples in the query result $\sem{Q}(D)$;
this takes time $O(|\ciD|)$ (in data complexity).
To the best of our knowledge, the color-index is the first index structure for fc-ACQs
that allows efficient enumeration and counting with performance guarantees
that may be sublinear in the database size.

The starting point for constructing the color-index is to represent $D$ by
a suitable vertex-labeled and edge-labeled directed graph,
to which we then apply a variant of the well-known \emph{color refinement} algorithm.
Color refinement is a simple and widely used subroutine for graph isomorphism testing algorithms;
see e.g.~\cite{BBG-ColorRefinement,grohe_color_2021} for an overview
and~\cite{CFI-paper,Arvind2017,Kiefer2020,Kiefer2021} for details on its expressibility.
Its result is a particular coloring of the elements in the active domain of $D$.
The construction of the color-index and, in particular,
the color database $\ciD$ is based on this coloring.
Our result relies on a close connection between the colors computed by the
color refinement algorithm and the homomorphisms from acyclic conjunctive queries to the database.
This connection between colors and homomorphisms from tree-like structures has been used before
in different contexts~\cite{Dvorak2010,grohe_dimension_2014,Grohe2020a,Bollen2023,Kayali2022}.
Notions of index structures that are based on concepts of bisimulations
(which produce results similar to color refinement)
and geared towards conjunctive query evaluation have been proposed and empirically evaluated,
e.g., in~\cite{Picalausa2014,Picalausa2012,Tran2013}.
But to the best of our knowledge, the present paper is the first to use color refinement
to produce an index structure that guarantees efficient constant-delay enumeration and counting.

The remainder of the paper is organized as follows.
Section~\ref{sec:preliminaries} fixes the basic notation concerning databases and queries.
Section~\ref{sec:indexing} formalizes the general setting of \emph{indexing for query evaluation}.
Section~\ref{sec:ACQs} provides the background on fc-ACQs that is necessary
for achieving our main result.
Section~\ref{sec:mainresult} gives a detailed description of the color-index,
while Section~\ref{sec:eval} shows how it can be utilized for
enumerating and counting the results of fc-ACQs.
Section~\ref{sec:MainTheorem} provides a precise statement of our main result.
We conclude in Section~\ref{sec:conclusion} with a summary and an outlook on future work.
Further details are provided an appendix.

\section{Preliminaries}\label{sec:preliminaries}

We write $\NN$ for the set of non-negative integers, and we let
$\NNpos\deff\NN\setminus\set{0}$ and $[n]\deff\set{1,\ldots,n}$ for all $n\in\NNpos$.
Whenever $G$ denotes a graph (directed or undirected),
we write $V(G)$ and $E(G)$ for the set of nodes and the set of edges of $G$.
Given a set $U \subseteq V(G)$, the subgraph of $G$ \emph{induced} by $U$
(for short: $G[U]$) is the graph $G'$ such that $V(G') = U$ and $E(G')$ is
the set of all $e\in E(G)$ such that both endpoints of $e$ belong to $U$.
A \emph{connected component} of an undirected graph is a maximal connected subgraph of $G$.
A \emph{forest} is an undirected acyclic graph; and a \emph{tree} is a connected~forest.

We usually write $\bar{a} =(a_1,\ldots,a_k)$ to denote a $k$-tuple
(for some arity $k\in\NN$) and write $a_i$ to denote its $i$-th component.
Note that there is only one tuple of arity 0, namely,
the \emph{empty tuple} denoted as $\emptytuple$.
For a $k$-tuple $\bar{a}$ and an $\ell$-tuple $\bar{c}$,
we write $\bar{a}{\cdot}\bar{c}$ or $\bar{a}\bar{c}$ for
the $(k{+}\ell)$-tuple $({a}_1,\ldots,{a}_k,{c}_1,\ldots,{c}_\ell)$.
Given a function $f\colon X\to Y$ and a $k$-tuple $\bar{x}$ of elements in $X$,
we write $f(\bar{x})$ for the $k$-tuple $(f({x}_1),\ldots,f({x}_k))$.
For $S\subseteq X^k$ we let $f(S)=\setc{f(\bar{x})}{\bar{x}\in S}$.

We fix a countably infinite set $\Dom$ for the \emph{domain} of potential database entries,
which we also call \emph{constants}.
A \emph{schema} $\sigma$ is a finite, non-empty set of relation symbols,
where each $R\in\sigma$ is equipped with a fixed arity $\ar(R)\in\NNpos$.
A schema $\sigma$ is called \emph{binary} if every $R\in\sigma$ has arity $\ar(R)\leq 2$.

A \emph{database} $D$ of schema $\sigma$ ($\sigma$-db, for short)
is a tuple of the form $D=(R^D)_{R\in\sigma}$, where $R^D$ is a finite subset of $\Dom^{\ar(R)}$. %
The \emph{active domain} $\adom{D}$ of $D$ is the smallest subset $S$ of $\Dom$
such that $R^D\subseteq S^{\ar(R)}$ for all $R\in \sigma$.
The \emph{size} $|D|$ of $D$ is defined as the total number of tuples in $D$, namely,
$|D|=\sum_{R\in\sigma}|R^D|$.

A $k$-ary \emph{query} (for $k\in\NN$) of schema $\sigma$ ($\sigma$-query, for short) is a
syntactic object $Q$ which is associated with a function $\sem{Q}$ that maps every $\sigma$-db $D$
to a finite subset of $\Dom^k$, and which is \emph{generic} in the following sense.
For every permutation $\pi$ of $\Dom$ and every $\sigma$-db $D$ we have
$\pi(\sem{Q}(D)) = \sem{Q}(\pi(D))$, where $\pi(D)$ is the $\sigma$-db $D'$
with $R^{D'}\deff \pi(R^D)$ for all $R\in\sigma$.
\emph{Boolean} queries are $k$-ary queries for $k=0$ and
\emph{non-Boolean} queries are $k$-ary queries with $k\geq 1$.
Note that there are only two relations of arity 0, namely, $\emptyset$ and $\set{\emptytuple}$.
For a Boolean query $Q$ we write $\sem{Q}(D)=\Yes$ to indicate that $\emptytuple\in\sem{Q}(D)$,
and we write $\sem{Q}(D)=\No$ to indicate that~$\sem{Q}(D)=\emptyset$.
In this paper we will focus on the following evaluation tasks for a given $\sigma$-db $D$:
\begin{itemize}
	\item
	\textbf{Boolean query evaluation:} On input of a Boolean $\sigma$-query $Q$,
	decide if $\sem{Q}(D)=\Yes$.
	\item
	\textbf{Non-Boolean query evaluation:} Upon input of a non-Boolean $\sigma$-query $Q$,
	compute the relation $\sem{Q}(D)$.
	\item
	\textbf{Counting query evaluation:} Upon input of a $\sigma$-query $Q$,
	compute the number $|\sem{Q}(D)|$ of tuples in the result of $Q$ on $D$.
\end{itemize}

Concerning the second task, we are mainly interested in finding an \emph{enumeration algorithm}
for computing the tuples in $\sem{Q}(D)$.
Such an algorithm consists of two phases:
the \emph{preprocessing phase} and the \emph{enumeration phase}.
In the preprocessing phase, the algorithm is allowed to do arbitrary preprocessing to build
a suitable data structure $\DataStructure{D}{Q}$.
In the enumeration phase, the algorithm can use $\DataStructure{D}{Q}$ to enumerate all tuples
in $\sem{Q}(D)$ followed by an End-Of-Enumeration message $\EOE$.
We require here that each tuple is enumerated exactly once (i.e., without repetitions).
The \emph{delay} is the maximum time that passes between the start of the enumeration phase and
the first output, between the output of two consecutive tuples, and between the last tuple and $\EOE$.

For our algorithms, we will assume the RAM-model with a uniform cost measure.
In particular, storing and accessing elements of $\Dom$ requires $O(1)$ space and time.
This assumption implies that, for any $r$-ary relation $R^D$, we can construct
in time $O(r\cdot |R^D|)$ an index that allows to enumerate $R^D$ with $O(1)$ delay
and to test for a given $r$-tuple $\bar{c}$ whether $\bar{c}\in R^D$ in time $O(r)$.
Furthermore, this implies that given any finite partial function $f\colon A \to B$,
we can build a \emph{lookup table} in time $O(|\dom{f}|)$, where
$\dom{f} \isdef \set{ x \in A \mid f(x) \text{ is defined} }$,
and have access to $f(a)$ in constant time.
 
\section{Indexing for Query Evaluation}\label{sec:indexing}
In this section, we formalize the setting for \emph{indexing for query evaluation} for the tasks
of Boolean ($\booltask$), non-Boolean ($\enumtask$), and counting ($\counttask$) query evaluation
for a given class $\QueryClass$ of queries over a fixed schema $\sigma$.
We present here the general setting.
Later, we will instantiate it for a specific class of queries.
The scenario is:
\begin{me}
	\item
	As input we receive a $\sigma$-db $D$.
	We perform an \emph{indexing phase} in order to build a suitable data structure $\DSD$.
	This data structure shall be helpful later on to efficiently evaluate
	\emph{any} query $Q\in\QueryClass$.
	\item
	In the \emph{evaluation phase} we have access to $\DSD$.
	As input, we  receive arbitrary queries $Q\in\QueryClass$ and
	one of the three task descriptions $\booltask$, $\counttask$, or $\enumtask$,
	where $\booltask$ is only allowed in case that $Q$ is a Boolean query.
	The goal is to solve this query evaluation task for $Q$ on $D$.
\end{me}

This scenario resembles what happens in real-world database systems,
where indexes are built to ensure efficient access to the information stored in the database,
and subsequently these indexes are used for evaluating various input queries.
We formalize the problem as:
\begin{center}
	\framebox{
	\begin{tabular}{rl}
		\textbf{Problem:}
			\!\!\!\!\!\!&
			$\IndexingProblemGeneral$
			\\\hline\vspace{-3mm}\\
		\textbf{Indexing:}
			\!\!\!\!\!\!&
			$\left\{\text{
				\begin{tabular}{rl}
					\textbf{input:}
						&\!\!\!
						a $\sigma$-db $D$\\
					\textbf{result:}
						&\!\!\!
						a data structure $\DSD$
				\end{tabular}
			}\right.$%
			\\[3.5mm]\hline\vspace{-3mm}\\
		\!\!
		\textbf{Evaluation:}
			\!\!\!\!\!\!&
			$\left\{\!\!\!
				\text{
				\begin{tabular}{rl}
					\textbf{input:}
						&\!\!\!
						\parbox[t]{.615\linewidth}{
							a $\sigma$-query $Q\in\QueryClass$, and a task
							description in $\set{\booltask,\enumtask,\counttask}$
						}\\
					\textbf{output:}
						&\!\!\!
						\parbox[t]{.615\linewidth}{
							the correct answer solving the given task for $\sem{Q}(D)$
						}
				\end{tabular}
				}
			\right.$
			\!\!\!\!
	\end{tabular}
	}
\end{center}
\smallskip

The \emph{indexing time} is the time used for building the data structure $\DSD$.
It only depends on the input database $D$, and we usually measure it by a function
$\indexingtimefunc(D)$ that provides an upper bound on the time it takes
to perform the indexing phase for $D$.

The time it takes to answer a Boolean query $Q$ on $D$ or for counting the number of result tuples
of a query $Q$ on $D$ depends on~$Q$ and the particular properties of the data structure $\DSD$.
We usually measure these times by functions $\booltimefunc(Q,D)$ and
$\countingtimefunc(Q,D)$ that provide an upper bound on the time it takes
to perform the task by utilizing the data structure $\DSD$.
Concerning the task $\enumtask$, we measure the preprocessing time and the delay by functions
$\preprocessingtimefunc(Q,D)$ and $\delaytimefunc(Q,D)$ that provide upper bounds on the time taken
for preprocessing and the delay, respectively,
when using the data structure $\DSD$ to enumerate $\sem{Q}(D)$.

Note that for measuring running times we use $D$ and not $|D|$ (e.g., $\booltimefunc(Q,D)$),
because we want to allow the running time analysis to be more fine-grained
than just depending on the number of tuples of $D$.
The main result of this paper is a solution for $\IndexingProblemGeneral$ where $\QueryClass$ is
the class $\fcACQ$ of all \emph{free-connex acyclic conjunctive queries} of a
\emph{binary} schema $\sigma$ (i.e., all $R\in\sigma$ have arity $\ar(R)\leq 2$).
 
\section{Free-Connex Acyclic CQs}\label{sec:ACQs}
This section provides the necessary definitions and known results concerning $\fcACQ$.
We fix a countably infinite set $\Var$ of \emph{variables} such that $\Var\cap\Dom=\emptyset$.
An \emph{atom} $\myatom$ of schema $\sigma$ is of the form $R(x_1,\ldots,x_r)$ where $R\in\sigma$,
$r=\ar(R)$, and $x_1,\ldots,x_r\in\Var$.
We write $\vars(\myatom)$ for the set of variables occurring in $\myatom$.
Let $k\in\NN$.
A $k$-ary \emph{conjunctive query} (CQ) of schema $\sigma$ is of the form
\begin{equation}\label{eq:generalCQ}
	\Ans(z_1,\ldots,z_k) \leftarrow \myatom_1,\ldots,\myatom_d
\end{equation}
where $d\in\NNpos$, $\myatom_j$ is an atom of schema $\sigma$ for every $j\in [d]$,
and $z_1,\ldots,z_k$ are $k$ pairwise distinct variables in $\bigcup_{j\in[d]}\vars(\myatom_j)$.
The expression to the left (right) of $\leftarrow$ is called
the \emph{head} (\emph{body}) of the query.
For a CQ $Q$ of the form~\eqref{eq:generalCQ} we let
$\atoms(Q)=\set{\myatom_1,\ldots,\myatom_d}$, $\vars(Q)=\bigcup_{j\in[d]}\vars(\myatom_j)$,
and $\free(Q)=\set{z_1,\ldots,z_k}$.
The  \emph{(existentially) quantified} variables are the elements in
$\quant(Q)\deff\vars(Q)\setminus\free(Q)$.
A CQ $Q$ is called \emph{Boolean} if $\free(Q)=\emptyset$, and it is called \emph{full}
(or, \emph{quantifier-free}) if $\quant(Q)=\emptyset$.
The \emph{size} $|Q|$ of the query is defined as $|\atoms(Q)|$.

The semantics are defined as usual.
A \emph{valuation} $\val$ for $Q$ is a mapping $\val\colon \vars(Q)\to \Dom$.
A valuation $\val$ is a \emph{homomorphism} from $Q$ to a $\sigma$-db $D$ if for every atom
$R(x_1,\ldots,x_r)\in\atoms(Q)$ we have $(\val(x_1),\ldots,\val(x_r))\in R^D$.
The \emph{query result}  of a CQ $Q$ of the form~\eqref{eq:generalCQ} on
the $\sigma$-DB $D$ is defined as the set of tuples
\begin{equation}
	\sem{Q}(D) \;\isdef\; %
		\setc{\, (\val(z_1),\ldots,\val(z_k)) }
		{ \val \text{ is	a homomorphism from $Q$ to $D$}\,}.
\end{equation}

The \emph{hypergraph} $H(Q)$ of a CQ $Q$ is defined as follows.
Its vertex set is $\vars(Q)$, and it contains a hyperedge $\vars(\myatom)$
for every $\myatom\in\atoms(Q)$.
The \emph{Gaifman graph} $G(Q)$ of $Q$ is the undirected simple graph with vertex set $\vars(Q)$,
and it contains the edge $\set{x,y}$ whenever $x,y$ are distinct variables such that
$x,y\in\vars(\myatom)$ for some $\myatom\in\atoms(Q)$.

\emph{Acyclic} CQs and \emph{free-connex acyclic} CQs are standard notions studied in
the database theory literature
(cf.~\cite{DBLP:journals/jacm/BeeriFMY83,
	DBLP:journals/siamcomp/BernsteinG81,
	DBLP:journals/jcss/GottlobLS02,
	Bagan.2007,
	AHV-Book};
see~\cite{BGS-tutorial} for an overview).
A CQ $Q$ is called \emph{acyclic} if its hypergraph $H(Q)$ is \emph{$\alpha$-acyclic}, i.e.,
there exists an undirected tree $T=(V(T),E(T))$ (called a \emph{join-tree} of $H(Q)$ and of $Q$)
whose set of nodes $V(T)$ is precisely the set of hyperedges of $H(Q)$, and where for each variable
$x\in\vars(Q)$ the set $\setc{t\in V(T)}{x\in t}$ induces a connected subtree of $T$.
A CQ $Q$ is \emph{free-connex acyclic} if it is acyclic \emph{and} the hypergraph obtained
from $H(Q)$ by adding the hyperedge $\free(Q)$ is $\alpha$-acyclic.

Note that any CQ $Q$ that is either Boolean or full is free-connex acyclic iff it is acyclic.
However, $\Ans(x,z)\leftarrow R(x,y), R(y,z)$ is an example of a query that is acyclic,
but not free-connex acyclic.

For the special case of a \emph{binary} schema, there is a particularly simple characterization of
(free-connex) acyclic CQs that will be useful for our purpose
(see Appendix~\ref{appendix:fcACQ} for a proof).
\begin{proposition}[Folklore]\label{prop:binaryfcacqs}\upshape
	A CQ $Q$ of a \emph{binary} schema $\sigma$ is \emph{acyclic}
	iff its Gaifman graph $G(Q)$ is acyclic.
	The CQ $Q$ is \emph{free-connex acyclic} iff $G(Q)$ is acyclic
	and the following statement is true:\
	For every connected component $C$ of $G(Q)$,
	the subgraph of $C$ induced by the set $\free(Q) \cap V(C)$ is connected or empty.
\end{proposition}
We write $\fcACQ$ to denote the set of all free-connex acyclic CQs of schema $\sigma$.
In the following, we state two seminal results by Yannakakis~\cite{Yannakakis1981} and Bagan,
Durand, and Grandjean~\cite{Bagan.2007} that will be crucial for the main result of this paper.

\begin{theorem}[Yannakakis~\cite{Yannakakis1981}]\label{thm:Yannakakis}
	For every schema $\sigma$ there is an algorithm that receives as input a $\sigma$-db $D$ and
	a Boolean acyclic CQ $Q$ of schema $\sigma$ and takes time
	$O(|Q|{\cdot} |D|)$ to decide if $\sem{Q}(D)=\Yes$.
\end{theorem}

In the literature, Yannakakis' result is often described in terms of \emph{data complexity}, i.e.,
running time components that depend on the query are hidden in the O-notation.
Adopting this view, Yannakakis' result states that Boolean acyclic CQs can be evaluated with
linear-time data complexity.
Bagan, Durand, and Grandjean~\cite{Bagan.2007} achieved an influential result for the task of
enumerating the set $\sem{Q}(D)$ for (non-Boolean) \emph{free-connex acyclic} CQs $Q$:
They found an enumeration algorithm with $O(|D|)$ preprocessing time and constant delay.
This statement refers to data complexity, i.e.,
running time components that depend on the query are hidden in the O-notation.
Several proofs of (and different algorithms for) Bagan, Durand and Grandjean's theorem
are available in the literature~\cite{Bagan.2007,
	Bagan_PhD,
	BraultBaron_PhD,
	DBLP:journals/tods/OlteanuZ15,
	DynamicYannakakis2017,
	DBLP:journals/pvldb/IdrisUVVL18,
	DBLP:journals/sigmod/IdrisUVVL19};
all of them focus on data complexity.
For this paper we need a more refined statement that takes into account the combined complexity of
the problem, and that is implicit in~\cite{BGS-tutorial} (see Appendix~\ref{appendix:fcACQ}).

\begin{theorem}\label{thm:BGS-enum}
	For every schema $\sigma$ there is an enumeration algorithm that receives as input
	a $\sigma$-db $D$ and a query $Q\in\fcACQ$ and that computes within preprocessing time
	$O(|Q|{\cdot}|D|)$ a data structure for enumerating $\sem{Q}(D)$ with delay $O(|\free(Q)|)$.
\end{theorem}

Our main result provides a solution for the problem $\IndexingProblem{\sigma}{\fcACQ}$
for \emph{binary} schemas $\sigma$.
The data structure $\DSD$ that we build for a given database $D$ during the indexing phase will
provide a new database $\ciD$.
This database $\ciD$ is potentially much smaller than $D$, and it will allow us to improve
the preprocessing time provided by Theorem~\ref{thm:BGS-enum} to $O(|Q|{\cdot}|\ciD|)$.
At a first glance one may be tempted to believe that it should be straightforward to generalize our
result from binary to arbitrary schemas by representing databases of a non-binary schema $\sigma$
as databases of a suitably chosen binary schema $\sigma'$.
However, the notion of fc-ACQs is quite subtle, and an fc-ACQ of schema $\sigma$ might translate
into a CQ of schema $\sigma'$ that is not an fc-ACQ
(see Appendix~\ref{appendix:BinaryToArbitrary} for an example), hence prohibiting to
straightforwardly reduce the case of general schemas to our results for fc-ACQs over binary schemas.  
\section{Construction of the color-index}\label{sec:mainresult}
For the remainder of this paper let us fix an arbitrary \emph{binary} schema $\sigma$,
i.e., every $R\in\sigma$ has arity $\ar(R)\leq 2$.
Our main result is a solution for $\IndexingProblemOurs$.
In this section, we describe the indexing phase of our solution.
Upon input of a $\sigma$-db $D$, we will build a data structure $\DSD$
that we call the \emph{color-index} for $\fcACQ$.
To describe this data structure, we need the following~notions.

\paragraph{Graph representation}
Let $D$ be a $\sigma$-db that we receive as input.
We start by encoding self-loops in $D$ into unary predicates.
Let $\sigmaOne$ be the schema obtained by inserting into $\sigma$ a new unary relation symbol
$S_R$ for every $R \in\sigma$ with $\ar(R)=2$.
We turn $D$ into a $\sigmaOne$-db $\DOne$ by letting $R^{\DOne}=R^D$ for every $R\in\sigma$,
and $S_R^{\DOne} \deff \setc{(v)}{(v,v)\in R^D}$ for every $R\in\sigma$ with $\ar(R)=2$.
Thus, relation $S_R$ of $\DOne$ consists of all elements $v$ such that $(v,v)\in R^D$
(i.e., $(v,v)$ is a \emph{self-loop} w.r.t.\ $R$ in~$D$).
Note that self-loops are still present in each binary relation $R$,
and $S_R$ is a redundant representation of these self-loops.
Clearly, the size of $\DOne$ is linear in the size of $D$.

We represent $\DOne$ as a vertex-labeled and edge-labeled directed graph
$\GOne=(\VOne,\EOne,\vl,\el)$ as follows:
$\VOne=\Adom(D)$ and $\EOne$ consists of all directed edges $(v,w)$ such that $v\neq w$ and
$(v,w)\in R^{\DOne}$ or $(w,v)\in R^{\DOne}$ for some $R\in \sigmaOne$ with $\ar(R)=2$.
Further, we define the vertex- and edge-labeling functions $\vl$ and $\el$:
\begin{alignat*}{3}
	\vl(v) \;\;&\deff\;\;\;&& 
		\setc{\,U\in\sigmaOne}{\ar(U)=1,\ (v)\in U^{\DOne}\,} \smallskip\\
	\el(v,w) \;\;&\deff&& 
		\setc{\,(R,\MyPlus)}{R\in\sigmaOne,\ \ar(R)=2,\ (v,w)\in R^{\DOne}\,} \ \,\cup\\
		& &&
		\setc{\,(R,\MyMinus)}{R\in\sigmaOne,\ \ar(R)=2,\ (w,v)\in R^{\DOne}\,}
\end{alignat*}
for every vertex $v\in\VOne$ and every edge $(v,w)\in\EOne$.

Note that $(\VOne,\EOne)$ is a self-loop-free simple and symmetric directed graph.
Self-loops present in $D$ are not represented by (labeled) edges but by vertex-labels in $\GOne$
(to enable this, we moved over from $D$ to $\DOne$).
For every edge $(v,w)\in \EOne$, the label $\el(v,w)$ determines the label $\el(w,v)$ of
the corresponding ``backwards edge'' as follows.
If $\el(v,w)=\lambda$, then $\el(w,v)=\Dual{\lambda}$ for the \emph{dual} label
$\Dual{\lambda}\deff \setc{(R,\MyPlus)}{(R,\MyMinus)\in \lambda}
\cup\setc{(R,\MyMinus)}{(R,\MyPlus)\in \lambda}$.

Let $\VLabels\deff \img(\vl)$ and $\ELabels\deff\img(\el)$ be the sets of labels that are
actually used as vertex labels and edge labels, respectively, of~$\GOne$.
One can easily check that $\GOne$ can be generated in time $O(|\DOne|)$, and $|\EOne|=O(|\DOne|)$.
Furthermore, although the number of potential labels of $\vl$ and $\el$
could be exponential in $|\sigmaOne|$, it holds that the number of actual labels
$|\VLabels|$ and $|\ELabels|$ is at most linear in $|\DOne|$.

\begin{exampleWithEndmarker}\label{example:running1}
	Consider the following running example about movies and actors
	taken almost verbatim from~\cite[Fig. 2]{AnglesABHRV17}.
	The schema $\exsigma$ has four binary relations symbols
	\exrel{Plays}, \exrel{ActedBy}, \exrel{Movie}, and \exrel{Screentime}
	(denoted by \exrel{P}, \exrel{A}, \exrel{M}, and \exrel{S}),
	and the database $\exdb$ has the following relations and tuples:
	\smallskip%
	\begin{center}%
	\begin{tabular}{cccc}
		\begin{tabular}{r|cc}
			\exrel{P} & &  \\ \hline
			& \exdata{PS} & \exdata{LM} \\
			& \exdata{PS} & \exdata{MM}
		\end{tabular} \qquad
		&
		\begin{tabular}{r|cc}
			\exrel{A} & &  \\ \hline
			& \exdata{LM} & \exdata{PS} \\
			& \exdata{MM} & \exdata{PS}
		\end{tabular} \qquad
		&
		\begin{tabular}{r|cc}
			\exrel{M} & &  \\ \hline
			& \exdata{LM} & \exdata{Dr.S} \\
			& \exdata{MM} & \exdata{Dr.S} 
		\end{tabular} \qquad
		&
		\begin{tabular}{r|cc}
			\exrel{S} & &  \\ \hline
			& \exdata{LM} & \exdata{18m} \\
			& \exdata{MM} & \exdata{34m}
		\end{tabular}
	\end{tabular}
	\end{center}
	\smallskip
	where Peter Sellers (\exdata{PS}) is an actor who plays as Lionel Mandrake (\exdata{LM}) 
	and Merkin Muffley (\exdata{MM}) in the same movie ``Dr.\ Strangelove'' (\exdata{Dr.S}).
	Each character appears 18 minutes (\exdata{18m}) and 34 minutes (\exdata{34m}) 
	in the movie, respectively.
	In Figure~\ref{fig:example:a}, we show its corresponding vertex-labeled and edge-labeled
	directed graph $\exgraph$.
\end{exampleWithEndmarker}

\begin{figure*}[t]
	\centering
	\captionsetup[subfigure]{justification=centering}
	\begin{subfigure}[t]{0.38\textwidth}
		\centering
		\begin{tikzpicture}[
				defaultstyle,
				level distance=2.5cm,
				sibling distance=2cm,
				mnode/.style={draw, rounded corners, minimum width=6mm},
				edge from parent/.style={draw=none},
				enode/.style={inner sep=1mm, sloped, at start, font={\tiny }},
				medge/.style={->},
			]
			\node[mnode] (PS) at (0,0) {\exdata{PS}}
			child {
				[level distance=1.5cm]
				node[mnode] (LM) {\exdata{LM}} 
				child { node[mnode] (18m) {\exdata{18m}} }
				child { node[mnode] (DrS) {\exdata{Dr.S}}}	
			}
			child {
				[level distance=1.5cm]
				node[mnode] (MM) {\exdata{MM}}
				child { (DrS) }
				child { node[mnode] (34m) {\exdata{34m}} }
			};

			\draw ([xshift=-1.5ex] PS.south) 
				edge[medge] 
				node[enode, anchor=south east] {$(\exrel{P}, \rightarrow), (\exrel{A}, \leftarrow)$} 
			([xshift=-.5ex] LM.north);
			\draw ([xshift=.5ex] LM.north) 
				edge[medge] 
				node[enode, anchor=north west] {$(\exrel{P}, \leftarrow), (\exrel{A}, \rightarrow)$} 
			([xshift=-.5ex] PS.south);
			\draw ([xshift=-.5ex] MM.north) 
				edge[medge] 
				node[enode, anchor=north east] {$(\exrel{P}, \leftarrow), (\exrel{A}, \rightarrow)$}
			([xshift=.5ex] PS.south);
			\draw ([xshift=1.5ex] PS.south) 
				edge[medge] 
				node[enode, anchor=south west] {$(\exrel{P}, \rightarrow), (\exrel{A}, \leftarrow)$}
			([xshift=.5ex] MM.north);
			\draw ([xshift=-1.5ex] DrS.north) 
				edge[medge] 
				node[enode, anchor=north east] {$(\exrel{M}, \leftarrow)$}
			([xshift=.5ex] LM.south);
			\draw ([xshift=1.5ex] LM.south) 
				edge[medge] 
				node[enode, anchor=south west] {$(\exrel{M}, \rightarrow)$}
			([xshift=-.5ex] DrS.north);
			\draw ([xshift=-1.5ex] MM.south) 
				edge[medge] 
				node[enode, anchor=south east] {$(\exrel{M}, \rightarrow)$} 
			([xshift=.5ex] DrS.north);
			\draw ([xshift=1.5ex] DrS.north) 
				edge[swap, medge] 
				node[enode, anchor=north west] {$(\exrel{M}, \leftarrow)$} 
			([xshift=-.5ex] MM.south);
			\draw ([xshift=-.5ex] 34m.north) 
				edge[medge]
				node[enode, anchor=north east] {$(\exrel{S}, \leftarrow)\,$} 
			([xshift=.5ex] MM.south);
			\draw ([xshift=1.5ex] MM.south) 
				edge[swap, medge] 
				node[enode, anchor=south west] {$(\exrel{S}, \rightarrow)$} 
			([xshift=.5ex] 34m.north);
			\draw ([xshift=-1.5ex] LM.south) 
				edge[medge] 
				node[enode, anchor=south east] {$(\exrel{S}, \rightarrow)$} 
			([xshift=-.5ex] 18m.north);
			\draw ([xshift=.5ex] 18m.north) 
				edge[swap, medge] 
				node[enode, anchor=north west] {$\,(\exrel{S}, \leftarrow)$} 
			([xshift=-.5ex] LM.south);
		\end{tikzpicture}
		\caption{}\label{fig:example:a}
	\end{subfigure}
	\hfill
	\begin{subfigure}[t]{0.38\textwidth}
		\centering
		\begin{tikzpicture}[
				defaultstyle,
				level distance=2.5cm,
				sibling distance=2cm,
				mnode/.style={draw, rounded corners, minimum width=6mm},
				edge from parent/.style={draw=none},
				enode/.style={inner sep=1mm, sloped, at start, font={\tiny }},
				medge/.style={->},
			]
			\node[mnode, fill=blue!40] (PS) at (0,0) {\exdata{B}}
			child {
				[level distance=1.5cm]
				node[mnode, fill=red!40] (LM) {\exdata{R}} 
				child { node[mnode, fill=yellow!40] (18m) {\exdata{Y}} }
				child { node[mnode, fill=green!40] (DrS) {\exdata{G}}}	
			}
			child {
				[level distance=1.5cm]
				node[mnode, fill=red!40] (MM) {\exdata{R}}
				child { (DrS) }
				child { node[mnode, fill=yellow!40] (34m) {\exdata{Y}} }
			};

			\draw ([xshift=-1.5ex] PS.south) 
				edge[medge] 
				node[enode, anchor=south east] {$(\exrel{P}, \rightarrow), (\exrel{A}, \leftarrow)$} 
			([xshift=-.5ex] LM.north);
			\draw ([xshift=.5ex] LM.north) 
				edge[medge] 
				node[enode, anchor=north west] {$(\exrel{P}, \leftarrow), (\exrel{A}, \rightarrow)$} 
			([xshift=-.5ex] PS.south);
			\draw ([xshift=-.5ex] MM.north) 
				edge[medge] 
				node[enode, anchor=north east] {$(\exrel{P}, \leftarrow), (\exrel{A}, \rightarrow)$}
			([xshift=.5ex] PS.south);
			\draw ([xshift=1.5ex] PS.south) 
				edge[medge] 
				node[enode, anchor=south west] {$(\exrel{P}, \rightarrow), (\exrel{A}, \leftarrow)$}
			([xshift=.5ex] MM.north);
			\draw ([xshift=-1.5ex] DrS.north) 
				edge[medge] 
				node[enode, anchor=north east] {$(\exrel{M}, \leftarrow)$}
			([xshift=.5ex] LM.south);
			\draw ([xshift=1.5ex] LM.south) 
				edge[medge] 
				node[enode, anchor=south west] {$(\exrel{M}, \rightarrow)$}
			([xshift=-.5ex] DrS.north);
			\draw ([xshift=-1.5ex] MM.south) 
				edge[medge] 
				node[enode, anchor=south east] {$(\exrel{M}, \rightarrow)$} 
			([xshift=.5ex] DrS.north);
			\draw ([xshift=1.5ex] DrS.north) 
				edge[swap, medge] 
				node[enode, anchor=north west] {$(\exrel{M}, \leftarrow)$} 
			([xshift=-.5ex] MM.south);
			\draw ([xshift=-.5ex] 34m.north) 
				edge[medge]
				node[enode, anchor=north east] {$(\exrel{S}, \leftarrow)$} 
			([xshift=.5ex] MM.south);
			\draw ([xshift=1.5ex] MM.south) 
				edge[swap, medge] 
				node[enode, anchor=south west] {$(\exrel{S}, \rightarrow)$} 
			([xshift=.5ex] 34m.north);
			\draw ([xshift=-1.5ex] LM.south) 
				edge[medge] 
				node[enode, anchor=south east] {$(\exrel{S}, \rightarrow)$} 
			([xshift=-.5ex] 18m.north);
			\draw ([xshift=.5ex] 18m.north) 
				edge[swap, medge] 
				node[enode, anchor=north west] {$(\exrel{S}, \leftarrow)$} 
			([xshift=-.5ex] LM.south);
		\end{tikzpicture}
		\caption{}\label{fig:example:b}
	\end{subfigure}
	\hfill 
	\begin{subfigure}[t]{0.2\textwidth}
		\centering
		\begin{tikzpicture}[
				defaultstyle,
				level distance=2.5cm,
				sibling distance=2cm,
				mnode/.style={draw, rounded corners, minimum width=6mm},
				edge from parent/.style={draw=none},
				enode/.style={inner sep=1mm, sloped, midway, font={\tiny }},
				medge/.style={->},
			]
			\node[mnode, fill=blue!40] (PS) at (0,0) {\exdata{B}}
			child {
				[level distance=1.5cm]
				node[mnode, fill=red!40] (LM) {\exdata{R}} 
				child { node[mnode, fill=yellow!40] (18m) {\exdata{Y}} }
				child { node[mnode, fill=green!40] (DrS) {\exdata{G}}}	
			};

			\draw ([xshift=-.5ex] PS.south) 
				edge[medge] 
				node[enode, anchor=north] {$(\exrel{P}, \rightarrow), (\exrel{A}, \leftarrow)$} 
			([xshift=-.5ex] LM.north);
			\draw ([xshift=.5ex] LM.north) 
				edge[medge] 
				node[enode, anchor=south] {$(\exrel{P}, \leftarrow), (\exrel{A}, \rightarrow)$} 
			([xshift=.5ex] PS.south);
			\draw ([xshift=-1.5ex] DrS.north) 
				edge[medge] 
				node[enode, anchor=north] {$(\exrel{M}, \leftarrow)$}
			([xshift=.5ex] LM.south);
			\draw ([xshift=1.5ex] LM.south) 
				edge[medge] 
				node[enode, anchor=south] {$(\exrel{M}, \rightarrow)$}
			([xshift=-.5ex] DrS.north);
			\draw ([xshift=-1.5ex] LM.south) 
				edge[medge] 
				node[enode, anchor=south] {$(\exrel{S}, \rightarrow)$} 
			([xshift=-.5ex] 18m.north);
			\draw ([xshift=.5ex] 18m.north) 
				edge[swap, medge] 
				node[enode, anchor=north] {$(\exrel{S}, \leftarrow)$} 
			([xshift=-.5ex] LM.south);
		\end{tikzpicture}
		\caption{}\label{fig:example:c}
	\end{subfigure}

	\caption{%
		(a) shows the graph $\exgraph$ representing the database $\exdb$ from
		Example~\ref{example:running1} (all vertices $v$ have the same label $\vl(v)=\emptyset$),
		(b) shows a coarsest stable coloring of $\exgraph$,
		and (c) shows a graph representing the color database $\exdbcol$ of $\exdb$.%
	}\label{fig:example}
	\vspace{-3mm}
\end{figure*}

\paragraph{Color refinement}
Once having produced $\GOne$, we apply to it a suitable variant of the well-known
\emph{color refinement} algorithm.
A high-level description of this algorithm, basically taken from~\cite{BBG-ColorRefinement},
is as follows.
The algorithm classifies the vertices of $\GOne$ by iteratively refining a coloring of the vertices.
Initially, each vertex $v$ has color $\vl(v)$.
Then, in each step of the iteration, two vertices $v,w$ that currently have the same color get
different refined colors if for some color $c$ of the currently used vertex-colors and
for some edge-label $\lambda\in\ELabels$,
we have $\numSucc{\lambda}{v}{c} \neq \numSucc{\lambda}{w}{c}$.
Here, for any vertex $u$ of $\GOne$, we let $\numSucc{\lambda}{u}{c}\deff |\NSucc{\lambda}{u}{c}|$
where $\NSucc{\lambda}{u}{c}$ denotes the set of all outgoing $\EOne$-neighbors $u'$ of $u$
(i.e., $(u,u')\in \EOne$) such that $\el(u,u') = \lambda$ and $u'$ has color $c$.
The process stops if no further refinement is achieved,
resulting in a \emph{stable coloring} of the vertices.

In~\cite{BBG-ColorRefinement}, Berkholz, Bonsma, and Grohe showed that this can be implemented
in such a way that it runs in time $O((|\VOne|+|\EOne|)\log |\VOne|)$.
To formally state their result, we need the following notation.
Let $f\colon \VOne \to S$ be a function, where $S$ is any set.
We say that a function $\col\colon \VOne \to C$ \emph{refines} $f$ iff for all
$v,w\in \VOne$ with $\col(v)=\col(w)$ we have $f(v)=f(w)$.
Further, we say that $\col$ \emph{is stable} iff for all $v,w\in \VOne$ with $\col(v)=\col(w)$,
and for every edge-label $\lambda\in \ELabels$ and for every color $c\in C$ we have:
$\numSucc{\lambda}{v}{c}=\numSucc{\lambda}{w}{c}$.
Finally, $\col$ is a \emph{coarsest stable coloring that refines $\vl$} iff $\col$ is stable,
$\col$ refines $\vl$, and for every coloring $\colAlt\colon\VOne\to C'$ (for some set $C'$)
that is stable and that refines $\vl$ we have: \,$\colAlt$ refines~$\col$.
The main result of~\cite{BBG-ColorRefinement} implies the following theorem
(see Appendix~\ref{appendix:proof_of_runtime} for details).

\begin{theorem}[\cite{BBG-ColorRefinement}]\label{thm:ColorRefinement}
  For any vertex-labeled and edge-labeled directed simple graph $\GOne=(\VOne,\EOne,\vl,\el)$, within time $O((|\VOne|+|\EOne|) \log|\VOne|)$ one can compute a coloring $\col\colon\VOne \to C$ (for a suitably chosen set $C$), such that $\col$ is a coarsest stable coloring that refines $\vl$.
\end{theorem}

In the indexing phase, we apply the algorithm provided
by Theorem~\ref{thm:ColorRefinement} to $\GOne$.
Let $\fcr(\DOne)$ denote the time taken for this, and note that
$\fcr(\DOne)\in O(|\DOne| \cdot\log |\Adom(\DOne)|)$.
Note that this is the \emph{worst-case} complexity; in case that $\DOne$ has a particularly simple
structure, the algorithm may terminate already in time $O(|\DOne|)$.

As a result we obtain a \emph{coarsest stable coloring $\col\colon\Adom(\DOne)\to C$
that refines $\vl$}, for a suitable set $C$.
Note that the number $|\img(\col)|$ of colors used by $\col$ is the smallest number possible
in order to obtain a stable coloring that refines $\vl$, and, moreover,
the coarsest stable coloring that refines $\vl$ is \emph{unique} up to a renaming of colors.

\begin{exampleWithEndmarker}\label{example:running2}
	Figure~\ref{fig:example:b} displays a stable coloring of the graph $\exgraph$
	(Figure~\ref{fig:example:a}) that corresponds to the database $\exdb$ presented in
	Example~\ref{example:running1}.
	For presentation, vertices are labeled with their corresponding colors
	(blue (\exdata{B}), red (\exdata{R}), green (\exdata{G}), yellow (\exdata{Y})).
	The reader can check that this is the coarsest stable coloring that refines $\vl$
	(in this case $\vl$ is trivial).
\end{exampleWithEndmarker}

Note that although nodes with the same color are ``isomorphic'' in the previous example
(in the sense that there is an isomorphism of the graph onto itself that maps a node of this color
to any other node of the same color), this is not always true.
One can see that in the graph
\raisebox{-0.8mm}{\mbox{
\begin{tikzpicture}[node distance=2mm,text width=5mm,
	wnode/.style={
		circle,
		inner sep=0pt,
		align=center,
		draw=black,
		scale=0.15,
		fill=white
	},
	bnode/.style={
		circle,
		inner sep=0pt,
		align=center,
		draw=black,
		scale=0.15,
		fill=black
	},
	stextnode/.style={
		circle,
		inner sep=-4pt,
		outer sep=0pt,
		draw=black,
		align=center,
		scale=1
	},
	mline/.style={-, line width=0.2pt}]
	\node (u1) at (0,0.13)  [wnode] {};
	\node (u2) at (0,-0.13) [wnode] {};
	\node (u3) at (0.2,0.13)  [bnode] {};
	\node (u4) at (0.2,-0.13) [bnode] {};
	\node (u5) at (0.4,0.13)  [wnode] {};
	\node (u6) at (0.4,-0.13) [wnode] {};
	
	\draw [mline] (u1) to (u2);
	\draw [mline] (u1) to (u3);
	\draw [mline] (u2) to (u4);
	\draw [mline] (u3) to (u4);
	\draw [mline] (u3) to (u5);
	\draw [mline] (u4) to (u6);
	\draw [mline] (u5) to (u6);
	
	\begin{scope}[xshift=6mm]
		\node (v1) at (0,0.13)  [wnode] {};
		\node (v2) at (0,-0.13) [wnode] {};
		\node (v3) at (0.2,0)  [bnode] {};
		\node (v4) at (0.4,0) [bnode] {};
		\node (v5) at (0.6,0.13)  [wnode] {};
		\node (v6) at (0.6,-0.13) [wnode] {};
		
		\draw [mline] (v1) to (v2);
		\draw [mline] (v1) to (v3);
		\draw [mline] (v2) to (v3);
		\draw [mline] (v3) to (v4);
		\draw [mline] (v4) to (v5);
		\draw [mline] (v4) to (v6);
		\draw [mline] (v5) to (v6);
	\end{scope}
\end{tikzpicture}}} 
all black nodes will receive the same color, although the black nodes on the left component are not
``isomorphic'' to the black nodes in the right component.
 
For later use let us summarize some crucial properties of $\col$.
Analogously to $\NSucc{\lambda}{u}{c}$ and $\numSucc{\lambda}{u}{c}$, let
$\numPred{\lambda}{u}{c}\deff |\NPred{\lambda}{u}{c}|$,
where $\NPred{\lambda}{u}{c}$ denotes all incoming $\EOne$-neighbors $u'$ of $u$
(i.e., $(u',u)\in \EOne$) such that $\el(u',u) = \lambda$ and $\col(u') = c$.

\begin{fact}\label{obs:col}\label{fact:col}
	For all $v,w\in\VOne$ with $\col(v)=\col(w)$, we have:\vspace{-\medskipamount}
	\begin{mea}
		\item\label{item:obs:col:a}
		$\vl(v)=\vl(w)$, and
		\item\label{item:obs:col:b}
		$\numSucc{\lambda}{v}{c}=\numSucc{\lambda}{w}{c}$ for all $\lambda\in\ELabels$
		and all $c\in C$, and
		\item\label{item:obs:col:c}
		$\numPred{\lambda}{v}{c}=\numPred{\lambda}{w}{c}$ for all $\lambda\in\ELabels$
		and all $c\in C$.
	\end{mea}
\end{fact}
\begin{proof}
	\eqref{item:obs:col:a} and~\eqref{item:obs:col:b} are met because $\col$ is %
	a refinement of $\vl$ and $\col$ is stable.
	\eqref{item:obs:col:c} follows from~\eqref{item:obs:col:b} because $\EOne$ is symmetric %
	and the label of a backward edge is the dual of the label of the corresponding forward edge,
	and hence $\NPred{\lambda}{v}{c}=\NSucc{\Dual{\lambda}}{v}{c}$.
\end{proof}

We now proceed to the final step of the indexing phase,
in which we use $\col$ to build the so-called \emph{color-index} data structure.

\paragraph{The color-index}
To describe the color-index structure $\DSD$, we use the following notation.
For all $c,c'\in C$ and every edge-label $\lambda$ we let
$\numSucc{\lambda}{c}{c'}\deff \numSucc{\lambda}{v}{c'}$
for some (and hence, for every, by Fact~\ref{fact:col}) $v\in \VOne$ with $\col(v)=c$.
For every potential edge-label $\lambda$
(i.e., $\lambda\subseteq\setc{(R,\MyPlus),\, (R,\MyMinus)}{R\in \sigmaOne,\, \ar(R)=2}$ with
$\lambda \neq \emptyset$) we let
\begin{alignat*}{3}
	\hatnumSucc{\lambda}{c}{c'} &\ \deff \
        \sum_{\lambda'\supseteq\lambda} \numSucc{\lambda'}{c}{c'}
        &&\qquad\text{and}\qquad &\hatNSucc{\lambda}{v}{c} &\ \deff \ \bigcup_{\lambda'\supseteq\lambda} \NSucc{\lambda'}{v}{c}\;,
\end{alignat*}
where the considered $\lambda'$ are actual edge-labels in $\ELabels$.
The reason why we exclude $\lambda=\emptyset$ is that we do not need it.
Including it would only unnecessarily blow up the number of tuples present in $\ciD$.
The data structure $\DSD$ that we build during the indexing phase consists of
the following five components:
\begin{me}
	\item\label{item:ciD:one}
	the schema $\sigmaOne$ and the $\sigmaOne$-db $\DOne$;
	\item\label{item:ciD:two}
	a lookup table to access the color $\col(v)$ given a vertex $v\in\VOne$, and
	an (inverse) lookup table to access the set $\setc{v\in\VOne}{\col(v) = c}$ given a $c\in C$,
	plus the number $\Numb{c}$ of elements in this set;
	\item\label{item:ciD:three}
	lookup tables to access the set $\hatNSucc{\lambda}{v}{c}$,
	given a potential edge-label $\lambda$, a vertex $v \in \VOne$, and a color $c\in C$;
	\item\label{item:ciD:four}
	lookup tables to access the number $\hatnumSucc{\lambda}{c}{c'}$,
	given a potential edge-label $\lambda$ and colors $c,c'\in C$;
	\item\label{item:ciD:five}
	the \emph{color database} $\ciD$ of the \emph{color schema} $\cisigma$ defined as follows:
	\begin{itemize}
		\item
		The schema $\cisigma$ consists of all the \emph{unary} relation symbols of $\sigmaOne$
		and a binary relation symbol $E_{\lambda}$ for every potential edge-label $\lambda$
		(i.e., $\lambda\subseteq\setc{(R,\MyPlus),\, (R,\MyMinus)}{R\in \sigmaOne,\, \ar(R)=2}$
		with~$\lambda\neq\emptyset$).
		\item The $\cisigma$-db $\ciD$ is defined as follows.
		For each unary relation symbol $U\in\sigmaOne$ we let $U^{\ciD}$ be the set of unary
		tuples $(c)$ for all $c\in C$ such that there is a $(v)\in U^{\DOne}$ with $\col(v)=c$.
		For every potential edge-label $\lambda$ we let $E_\lambda^{\ciD}$ be the set of
		all tuples $(c,c')\in C \times C$ such that $\hatnumSucc{\lambda}{c}{c'} > 0$.
		Note that $\Adom(\ciD)=C$.
	\end{itemize}
\end{me}

\begin{exampleWithEndmarker}\label{example:running3}
	For our example database $\exdb$ (Example~\ref{example:running1}),
	one can check that the corresponding color database $\exdbcol$ is:
	\begin{center}
	\begin{tabular}{rrr}
		\begin{tabular}{r|cc}
			$E_{\{(\exrel{P}, \rightarrow), (\exrel{A}, \leftarrow)\}}$ & &  \\ \hline
			& \exdata{B} & \exdata{R}
		\end{tabular}%
		&
		\begin{tabular}{r|cc}
			$E_{\{(\exrel{S}, \rightarrow)\}}$ & &  \\ \hline
			& \exdata{R} & \exdata{Y}
		\end{tabular}%
		&
		\begin{tabular}{r|cc}
			$E_{\{(\exrel{M}, \rightarrow)\}}$ & &  \\ \hline
			& \exdata{R} & \exdata{G}
		\end{tabular}\vspace{2mm}
          \\
		\begin{tabular}{r|cc}
			$E_{\{(\exrel{P}, \leftarrow), (\exrel{A}, \rightarrow)\}}$ & &  \\ \hline
			& \exdata{R} & \exdata{B}
		\end{tabular}%
		&
		\begin{tabular}{r|cc}
			$E_{\{(\exrel{S}, \leftarrow)\}}$ & &  \\ \hline
			& \exdata{Y} & \exdata{R}
		\end{tabular}%
		&
		\begin{tabular}{r|cc}
			$E_{\{(\exrel{M}, \leftarrow)\}}$ & &  \\ \hline
			& \exdata{G} & \exdata{R}
		\end{tabular}%
	\end{tabular}\medskip
\end{center}
\noindent
Further, 
$E^{\exdbcol}_{\set{(\exrel{P},\rightarrow)}}
= E^{\exdbcol}_{\set{(\exrel{A},\leftarrow)}}
= E^{\exdbcol}_{\set{(\exrel{P},\rightarrow),(\exrel{A},\leftarrow)}}$,
$E^{\exdbcol}_{\set{(\exrel{P},\leftarrow)}}
= E^{\exdbcol}_{\set{(\exrel{A},\rightarrow)}}
= E^{\exdbcol}_{\set{(\exrel{P},\leftarrow),(\exrel{A},\rightarrow)}}$
and \,$E^{\exdbcol}_{\lambda}=\emptyset$\, for all the remaining potential edge-labels $\lambda$.
In Figure~\ref{fig:example:c}, we present a graph representing the tuples of $\exdbcol$.
\end{exampleWithEndmarker}
 
Note that after having performed the color refinement algorithm (taking time $\fcr(D)$),
all the components~\eqref{item:ciD:one}--\eqref{item:ciD:five} of $\DSD$
can be built in total time $O(|D|)$.
(Recall from the problem statement in Section~\ref{sec:indexing} that
the schema $\sigma$ is fixed and not part of the input.
Thus, the exponential dependency on $\sigma$ caused by considering all
potential edge labels is suppressed by the O-notation.
In Appendix~\ref{appendix:MoreComplicatedIndexing} we outline how the exponential dependency
on $\sigma$ can be avoided.)
In summary, the indexing phase takes time $\fcr(D)+O(|D|)$.
The color-index database $\ciD$ has size $O(|D|)$ in the \emph{worst case}; but $|\ciD|$ might be
substantially smaller than $|D|$ in case that $D$ has a particularly well-behaved structure.
We formalize this statement in the following result.
\begin{proposition}\label{prop:colordbsize}
	There is a database family ${\{D^n\}}_{n\geq 3}$
	with $|D^n| \in \Theta(n)$ but $|\ciDn| \in O(1)$.
\end{proposition}
\begin{proof}
	Let $\sigma=\set{R}$ with $\ar(R)=2$.
	For $n\in\NN$ with $n\geq 3$ let $D^n$ be the $\sigma$-db $D$ with active domain $[n]$, where
	the relation $R^{D}$ consists of the tuple $(n,1)$ and the tuples $(i,i{+}1)$ for all $i<n$.
	I.e., $D$ is a directed cycle on $n$ nodes.
	Note that $\DOne=D$.
	The graph $\GOne=(\VOne,\EOne,\vl,\el)$ looks as follows:
	$\VOne=[n]$, $\EOne= R^{D}\cup\setc{(j,i)}{(i,j)\in R^{D}}$,
	$\vl(i)=\emptyset$ for every $i\in[n]$,
	and for every $(i,j)\in R^{D}$ we have $\el(i,j)=\set{(R,\rightarrow)}$ and
	$\el(j,i)=\set{(R,\leftarrow)}$.

	The color refinement algorithm will stop after time $\fcr(\DOne)=O(n)$ by assigning each node
	$i\in[n]$ the same color $\col(i)=c$, where $C=\set{c}$ consists of only one color
	(note that this is a coarsest stable coloring of $\GOne$ that refines $\vl$).

	Clearly, $\ELabels$ consists of only the edge-labels $\set{(R,\rightarrow)}$ and
	$\set{(R,\leftarrow)}$, and $\numSucc{\set{(R,\rightarrow)}}{c}{c}=1$ and
	$\numSucc{\set{(R,\leftarrow)}}{c}{c}=1$.
	Thus, $\hatnumSucc{\set{(R,\rightarrow)}}{c}{c}=1=\hatnumSucc{\set{(R,\leftarrow)}}{c}{c}$.
	Furthermore, for $\lambda_{\textit{max}} \deff\set{(R,\rightarrow), \allowbreak (R,\leftarrow)}$
	we have $\hatnumSucc{\lambda_{\textit{max}}}{c}{c}=0$.

	In summary, $\ciD$ is the database with active domain $C=\set{c}$, and with relations
	$E^{\ciD}_{\set{(R,\rightarrow)}} = E^{\ciD}_{\set{(R,\leftarrow)}} = \set{(c,c)}$ and
	$E^{\ciD}_{\lambda_{\textit{max}}}=\emptyset$.
	In particular, $|\ciD|=2$, while $|D|=n$.
\end{proof}

Note that we define the color-index (and, in particular, $\cisigma$ and $\ciD$) considering
the binary relation symbols $E_{\lambda}$ for every potential edge-label $\lambda$, i.e.,
$\lambda\subseteq\setc{\mbox{(R,\MyPlus)},\allowbreak (R,\MyMinus)}{R\in \sigmaOne,\, \ar(R)=2}$
with $\lambda\neq \emptyset$.
This implies an exponential blow-up on $|\sigma|$.
However, given that we assume that $\sigma$ is fixed,
this blow-up is a constant factor in the running time.
By spending some more effort, one can avoid the exponential blow-up on $|\sigma|$ entirely;
details on this can be found in Appendix~\ref{appendix:MoreComplicatedIndexing}.  

\section{Evaluation using the color-index}\label{sec:eval}
This section shows how to use the color-index $\DSD$ to evaluate
any free-connex acyclic CQ $Q$ over $D$. 
I.e., we describe the \emph{evaluation phase} of our solution for $\IndexingProblemOurs$.
We assume that the color-index $\DSD$ (including the color database $\ciD$)
of the given database $D$ has already been built during the indexing phase.
Let $Q\in\fcACQ$ be an arbitrary query that we may receive as an input during the evaluation phase.
We first explain how the evaluation tasks can be simplified by focusing on \emph{connected} queries.

\paragraph{Connected queries}
We say that a CQ $Q$ is \emph{connected} iff its Gaifman graph $G(Q)$ is connected.
Intuitively, if a CQ $Q$ is not connected, then we can write $Q$ as
\[
	\Ans(\bar{z}_1,\ldots,\bar{z}_\ell) \ \leftarrow \
	\bar{\alpha}_1(\bar{z}_1,\bar{y}_1),\ldots,\bar{\alpha}_\ell(\bar{z}_\ell,\bar{y}_\ell)
\]
such that each $\bar{\alpha}_i(\bar{z}_i,\bar{y}_i)$ is a sequence of atoms,
$\vars(\bar{\alpha}_i(\bar{z}_i,\bar{y}_i))$ and $\vars(\bar{\alpha}_j(\bar{z}_j,\bar{y}_j))$
are disjoint for $i \neq j$, and for each $i\in [\ell]$
the CQ $Q_i\deff \Ans(\bar{z}_i) \leftarrow \bar{\alpha}_i(\bar{z}_i,\bar{y}_i)$ is connected.
To simplify notation we assume w.l.o.g.\ that in the head of $Q$ the variables are ordered
in exactly the same way as in the list $\bar{z}_1,\ldots,\bar{z}_\ell$.
One can easily check that this decomposition satisfies
\begin{equation*}
	\sem{Q}(D) = \sem{Q_1}(D) \times \cdots \times \sem{Q_\ell}(D)
	\qquad\text{and}\qquad
	|\sem{Q}(D)| = |\sem{Q_1}(D)| \cdot  \ldots \cdot |\sem{Q_\ell}(D)|
\end{equation*}
for every $\sigma$-db $D$.
Then, we can compute $|\sem{Q}(D)|$ by first computing each number $|\sem{Q_i}(D)|$
and then multiplying all values in $O(|Q|)$-time.
Similarly, for enumerating $\sem{Q}(D)$ one can convert constant-delay enumeration algorithms
for $Q_1, \ldots, Q_\ell$ into one for $Q$ by iterating in nested loops over
the outputs of $\sem{Q_1}(D), \ldots,\allowbreak \sem{Q_\ell}(D)$.
By the previous argument, in the following we will concentrate on
the evaluation of \emph{connected} CQs without loss of generality.
For the remainder of this section, let us fix an arbitrary \emph{connected} CQ $Q\in\fcACQ$ and
assume that its head is of the form $\Ans(x_1,\ldots,x_k)$ (with $k\geq 0$).

\paragraph{The  modified query $\QOne$}
When receiving the connected query $Q\in\fcACQ$ with head $\Ans(x_1,\ldots,x_k)$,
the first step is to remove the self-loops of $Q$.
I.e., we translate $Q$ into the query $\QOne$ over schema $\sigmaOne$ as follows:
$\QOne$ is obtained from $Q$ by replacing every atom of the form $R(x,x)$ with the atom $S_R(x)$.
Obviously, the Gaifman graph of the query remains unchanged, i.e., $G(\QOne)=G(Q)$.
Clearly, $\QOne$ is a connected query in $\fcACQOne$ such that $\sem{\QOne}(\DOne)=\sem{Q}(D)$,
and $\QOne$ is \emph{self-loop-free}, i.e., every binary atom of $\QOne$
is of the form $R(x,y)$ with $x\neq y$ and $R\in\sigma$ with $\ar(R)=2$.

\paragraph{The color query $\ciQ$}
Next, we transform $\QOne$ into a query $\ciQ$ of schema $\cisigma$.
We start by picking an arbitrary variable $r$ in $\free(\QOne)$,
and if $\QOne$ is a Boolean query (i.e., $\free(\QOne) = \emptyset$),
then we pick $r$ as an arbitrary variable in $\vars(\QOne)$.
This $r$ will be fixed throughout the remainder of this~section.

Since $\QOne$ is connected and $\QOne\in \fcACQOne$,
the following is obtained from Proposition~\ref{prop:binaryfcacqs}:
Picking $r$ as the \emph{root node} turns the Gaifman graph $G(\QOne)$ into a rooted tree $T$.
Furthermore, there exists a strict linear order $<$ on $\vars(\QOne)$
that is compatible with the descendant relation of the rooted tree $T$
(i.e., for all $x,y\in\vars(\QOne)$ such that $x\neq y$ and $y$ is a descendant of $x$ in $T$,
we have $x<y$)
\emph{and} that satisfies $x<y$ for all variables $x\in \free(\QOne)$ and $y\in \quant(\QOne)$.
Such a $<$ can be obtained based on a variant of breadth-first search of $G(\QOne)$
that starts with the root node $r$ and that prioritizes free over quantified variables
using separate queues for free and for quantified variables.
We choose an arbitrary such order $<$, and we fix it throughout the remainder of this section.
We will henceforth assume that $r = x_1$ and that $x_1 < x_2 < \cdots < x_k$
(by reordering the variables $x_1,\ldots,x_k$ in the head of~$\QOne$ and $Q$,
this can be achieved without loss of generality).

Based on $\QOne$ and our choice of $<$,
we construct the \emph{color query} $\ciQ$ of schema $\cisigma$ by using the following procedure.
The query $\ciQ$ has exactly the same head as $\QOne$ (and $Q$).
We initialize $\atoms(\ciQ)$ to consist of all the \emph{unary} atoms of $\QOne$.
Afterwards we loop through all edges $e$ of the Gaifman graph $G(\QOne)$ and proceed as follows.
Let $e=\set{x,y}$ with $x<y$.
Initialize $\lambda_e$ to be the empty set.
For every $R\in \sigmaOne$ such that $R(x,y)\in\atoms(\QOne)$,
insert into $\lambda_e$ the tuple $(R,\MyPlus)$.
For every $R\in\sigma$ such that $R(y,x)\in\atoms(\QOne)$,
insert into $\lambda_e$ the tuple $(R,\MyMinus)$.
After having completed the construction of $\lambda_e$,
we insert into $\atoms(\ciQ)$ the atom $E_{\lambda_e}(x,y)$ (note that $\lambda_e\neq\emptyset$).

\begin{example}\label{example:queries}
	Let $\sigma=\set{R}$ with $\ar(R)=2$,
	and let $Q$ be the query
        \[
          \Ans(x_1,x_2) \leftarrow R(x_1,x_2), R(x_3,x_1), R(x_2,x_2).
        \]
	By Proposition~\ref{prop:binaryfcacqs} it is easy to see that $Q\in\fcACQ$.
	Note that $\sigmaOne=\set{R,S_R}$,
	and $\QOne$ is: \ \ $\Ans(x_1,x_2)\;\leftarrow \; R(x_1,x_2), \, R(x_3,x_1), \, S_R(x_2).$
	
	Choosing $x_1$ as the root of the Gaifman graph $G(\QOne)$, letting $x_1<x_2<x_3$, we get that
	$\lambda_{\set{x_1,x_2}}=\set{(R,\MyPlus)}$,\,
	and \,$\lambda_{\set{x_1,x_3}}=\set{(R,\MyMinus)}$.
	Furthermore, $\ciQ$ is the query:
	\[
	\Ans(x_1,x_2)\;\leftarrow\; S_R(x_2),\,
	E_{\set{(R,\MyPlus)}}(x_1,x_2),\, E_{\set{(R,\MyMinus)}}(x_1,x_3)\,. \tag*{$\lrcorner$}
	\]
\end{example}

Recall that when defining the color query $\ciQ$,
we defined labels $\lambda_e$ for every edge $e$ of the Gaifman graph $G(\QOne)$.
For every node $x$ of $G(\QOne)$, we let $\lambda_x$ be the set of all unary relation symbols
$U\in\sigmaOne$ such that $U(x)\in\Atoms(\QOne)$.
The following lemma summarizes the crucial correspondence between
homomorphisms from $\QOne$ to $\DOne$, the labels $\lambda_e$ and $\lambda_x$
associated to the edges $e$ and the nodes $x$ of the Gaifman graph of $\QOne$,
and the labeled graph $\GOne=(\VOne,\EOne,\vl,\el)$.

\begin{lemma}\label{lemma:homomorphisms}
	A mapping $\val\colon\vars(\QOne)\to \Dom$ is a homomorphism from $\QOne$ to $\DOne$
	iff the following is true:
	\vspace{-\medskipamount}%
	\begin{enumerate}[(1)]
		\item
		$\vl(\val(x))\supseteq \lambda_x$, for every $x\in\vars(\QOne)$, \ and
		\item
		$(\val(x),\val(y))\in \EOne$ and $\el(\val(x),\val(y))\supseteq \lambda_e$,
		for every edge $e=\set{x,y}$ of $G(\QOne)$ with $x<y$.
	\end{enumerate}
\end{lemma}
\begin{proof}
	This is an immediate consequence of our choice of $\GOne=(\VOne,\EOne,\vl,\el)$, $\QOne$,
	the sets $\lambda_x$ for all $x\in\vars(\QOne)$
	and the sets $\lambda_e$ for all edges $e$ of $G(\QOne)$.
\end{proof}

Note that the queries $Q$, $\QOne$, and $\ciQ$ have the same head and the same Gaifman graph.
Thus, as $Q\in\fcACQ$, also $\QOne\in\fcACQOne$ and $\ciQ\in\fcACQci$.
Upon input of $Q$ we can compute in time $O(|Q|)$ the rooted tree $T$,
the queries $\QOne$ and $\ciQ$, and lookup tables to obtain $O(1)$-time access to
$\lambda_x$ and $\lambda_e$ for each node $x$ and edge $e$ of the Gaifman graph of $Q$.

\paragraph{Evaluation phase for $\booltask$}
As a warm-up, we start with the evaluation of \emph{Boolean} queries.
For this task, we assume that $\QOne$ is a Boolean query, and as described before,
we assume w.l.o.g.\ that $\QOne$ is connected.
The following lemma shows that evaluating $\QOne$ on $\DOne$ can be reduced to evaluating
the color query $\ciQ$ on the color database $\ciD$.
\begin{lemma}\label{lemma:bool}
	If $\QOne$ is a Boolean query, then $\sem{\QOne}(\DOne) = \sem{\ciQ}(\ciD)$.
\end{lemma}
The proof easily follows from Lemma~\ref{lemma:homomorphisms}, our particular choice of
$\ciQ$ and $\ciD$, and the fact that $\col$ is a stable coloring of $\GOne$ that refines $\vl$
(see Appendix~\ref{sec:evalApp} for details).

Since $\ciQ$ is acyclic, we use the algorithm from Theorem~\ref{thm:Yannakakis} to compute
$\sem{\ciQ}(\ciD)$ in time $O(|\ciQ|{\cdot}|\ciD|)$.
From Lemma~\ref{lemma:bool} we know that $\sem{\QOne}(\DOne) =\sem{\ciQ}(\ciD)$.
In summary, by using the information produced during the indexing phase,
we can solve the task of evaluating a Boolean query $Q$ on $D$ in time $O(|Q|{\cdot}|\ciD|)$.
 
\paragraph{Evaluation phase for $\enumtask$}
We now move to the non-Boolean case.
We assume that $\QOne$ is a connected $k$-ary query for some $k \geq 1$.
Recall that the head of $\QOne$ and $\ciQ$ is of the form 
$\Ans(x_1,\ldots,x_k)$ with $x_1<\cdots < x_k$,
where $<$ is the fixed order we used for constructing the query $\ciQ$.
Further, recall that $T$ is the rooted tree obtained from the Gaifman graph $G(\QOne)$
by selecting $r=x_1$ as its root.
For each $1 < i \leq k$, there is a unique $j < i$ such that
$x_j$ is the parent of node $x_i$ in $T$;
we will henceforth write $p(x_{i})$ to denote this particular $x_j$.

The following technical lemma highlights the connection between 
$\sem{\QOne}(\DOne)$ and $\sem{\ciQ}(\ciD)$
that will allow us to use the color-index for enumerating $\sem{\QOne}(\DOne)$.

To state the lemma, we need the following notation.
For any $i\in [k]$, we say that $(v_1, \ldots, v_i)$ is a \emph{partial output} of $\QOne$
over $\DOne$ of color $(c_1,\ldots,c_k)$ iff there exists $(v_{i+1}, \ldots, v_k)$
such that $(v_1, \ldots, v_i,\allowbreak v_{i+1}, \ldots, v_k) \in \sem{\QOne}(\DOne)$
and $(\col(v_1),\ldots,\col(v_k))=(c_1,\ldots,c_k)$.  

\begin{lemma}\label{lemma:mainLemma}
	Suppose that the head of $Q$ (and $\QOne$ and $\ciQ$) is of the form
	$\Ans(x_1,\ldots,x_k)$ with $x_1 < \cdots < x_k$ and $k \geq 1$.
	The following statements are true:%
	\vspace{-\medskipamount}%
	\begin{mea}
		\item 
		For every $(v_1,\ldots,v_k)\in\sem{\QOne}(\DOne)$ we have
		$(\col(v_1),\ldots,\col(v_k))\in\sem{\ciQ}(\ciD)$.
		\item 
		For all $\ov{c}=(c_1,\ldots,c_k)\in\sem{\ciQ}(\ciD)$,
		and for every $v_1 \in \Adom(\DOne)$ with $\col(v_1) = c_1$, $(v_1)$
		is a partial output of $\QOne$ over $\DOne$ of color $\ov{c}$. 
		Moreover, if $(v_1,\ldots,v_i)$ is a partial output of $\QOne$ over $\DOne$
		of color $\ov{c}$ and $e=\set{x_j,x_{i+1}}$ is an edge of $G(\QOne)$
		with $x_j < x_{i+1}$, then $(v_1,\ldots,v_i,v_{i+1})$ is a partial output of
		$\QOne$ over $\DOne$ of color $\ov{c}$ for every
		$v_{i+1} \in \hatNSucc{\lambda_e}{v_j}{c_{i+1}}$ and,
		furthermore, $\hatNSucc{\lambda_e}{v_j}{c_{i+1}}\neq\emptyset$.
	\end{mea}  
\end{lemma}

The proof uses Lemma~\ref{lemma:homomorphisms}, our particular choice of $\ciQ$ and $\ciD$,
and the fact that $\col$ is a stable coloring of $\GOne$ that refines $\vl$
(see Appendix~\ref{sec:evalApp} for details).

For solving the task $\enumtask$, we proceed as follows:
We use Theorem~\ref{thm:BGS-enum} with input $\ciD$ and $\ciQ$ to carry out a
preprocessing phase in time $O(|\ciQ|{\cdot}|\ciD|)$
and then enumerate the tuples in $\sem{\ciQ}(\ciD)$ with delay $O(k)$.
Each time we receive a tuple $\ov{c}=(c_1,\ldots,c_k) \in \sem{\ciQ}(\ciD)$,
we carry out the following algorithm:

\begin{quote}
	for all $v_1\in \VOne$ with $\col(v_1)=c_1$ do\ $\myEnum((v_1);\ov{c})$
\end{quote}
\noindent
For all $i\in [k]$, the procedure $\myEnum((v_1,\ldots,v_i);\ov{c})$ is as follows:
\begin{quote}
	if $i=k$ then \textbf{output} $(v_1,\ldots,v_k)$  \\
	else
	\\
	\mbox{ \ \ \ }let $x_j \deff p(x_{i+1})$ and let $e\deff \set{x_j,x_{i+1}}$ \\
	\mbox{ \ \ \ }for all $v_{i+1}\in \hatNSucc{\lambda_e}{v_j}{c_{i+1}}$
	do \ 
	$\myEnum((v_1,\ldots,v_i,v_{i+1});\ov{c})$\\
	endelse.
\end{quote}

\noindent
Clearly, for each fixed $\ov{c}=(c_1,\ldots,c_k)\in\sem{\ciQ}(\ciD)$ this outputs,
without repetition, tuples $(v_1,\ldots,v_k)\in\VOne^k$ such that
$(\col(v_1),\ldots,\allowbreak \col(v_k))=\ov{c}$.
From Lemma~\ref{lemma:mainLemma}(b) we obtain that the output consists of
precisely those tuples $(v_1,\ldots,v_k)$ with color
$(\col(v_1),\ldots,\col(v_k))\allowbreak =\ov{c}$ that belong to $\sem{\QOne}(\DOne)$.
From Lemma~\ref{lemma:mainLemma}(a) we obtain that, in total,
the algorithm enumerates all the tuples $(v_1,\ldots,v_k)$ in $\sem{\QOne}(\DOne)$.

From Lemma~\ref{lemma:mainLemma}(b) we know that every time we encounter a loop
of the form ``{for all $v_{i+1} \in \hatNSucc{\lambda_e}{v_j}{c_{i+1}}$}'',
the set $\hatNSucc{\lambda_e}{v_j}{c_{i+1}}$ is non-empty.
Thus, by using the lookup tables built during the indexing phase,
we obtain that the delay between outputting any two tuples of $\sem{\QOne}(\DOne)$ is $O(k)$.

In summary, by using the information produced during the indexing phase,
we can enumerate $\sem{\QOne}(\DOne)$ with delay $O(k)$ after a
preprocessing phase that takes time $O(|Q|{\cdot}|\ciD|)$. 

\paragraph{Evaluation phase for $\counttask$}
In case that $\QOne$ is a \emph{Boolean} query, the number of answers is either 0 or 1,
and the result for the task $\counttask$ can be obtained by evaluating the Boolean query
as described previously.

In case that $\QOne$ is a $k$-ary query for some $k\geq 1$,
we make the same assumptions and use the same notation as for $\enumtask$.
For solving the task $\counttask$, we proceed as follows.
For every $c\in C$, let $v_c$ be the first vertex in the lookup table for the set
$\setc{v\in\VOne}{\col(v)=c}$
(this is part of the color-index $\DSD$ that has been produced during the indexing phase).
For every $x\in\vars(\QOne)$ and for every $c\in C$
let $f_1(c,x)\deff 1$ if $\vl(v_c)\supseteq \lambda_x$, and let $f_1(c,x)=0$ otherwise.
Note that $f_1(c,x)$ indicates whether one (and thus, every) vertex of color $c$
satisfies all the unary atoms of the form $U(x)$ that occur in $\QOne$.
By using the lookup tables built during the indexing phase and the lookup tables for $\lambda_x$ and
$\lambda_e$ built during the evaluation phase after having received the input query $Q$,
we can compute in time $O(|C|{\cdot}|\QOne|)$ a lookup table that gives us $O(1)$-time access to
$f_1(c,x)$ for all $c\in C$ and all $x\in\vars(\QOne)$.

Recall that $T$ denotes the rooted tree obtained from the Gaifman graph $G(\QOne)$
by choosing $x_1$ to be its root.
For every node $x$ of $T$ we let  $\Children(x)$ be the set of its children,
and we let $T_x$ be the subtree of $T$ rooted at $x$.
For $x\neq x_1$ we write $\Parent(x)$ to denote the parent of $x$ in $T$.

We perform a bottom-up pass of $T$ and fix the following notions:
For every $c\in C$ and for every \emph{leaf} $x$ of $T$, let:
\[
	f_{\downarrow}(c,x) \ \deff \ \ f_1(c,x).
\]
For every $c\in C$ and every node $y\neq x_1$ of $T$, for $e\deff \set{\Parent(y),y}$ let:
\[
	g(c,y) \ \deff \ \ \sum_{c'\in C} \ f_{\downarrow}(c',y)  \cdot \ \hatnumSucc{\lambda_e}{c}{c'}.
\]
For every $c\in C$ and for every node $x$ of $T$ that is not a leaf, let:
\[
	f_{\downarrow}(c,x)\ \deff \ \ f_1(c,x) \cdot\!\! \prod_{y\in\Children(x)} g(c,y).
\]

By induction (bottom-up along $T$) it is straightforward to prove the following lemma
(see Appendix~\ref{sec:evalApp}).

\begin{lemma}\label{lemma:counting}
	For every $c\in C$ and every $v\in\VOne$ with $\col(v)=c$,
	the following is true for all $x\in V(T)$:%
	\vspace{-\medskipamount}%
	\begin{mea}
		\item
		$f_{\downarrow}(c,x)$ is the number of mappings $\val\colon V(T_x) \to \VOne$
		satisfying $\val(x)=v$ and
		\begin{enumerate}[(1)]
			\item for every $x'\in V(T_x)$ we have $\vl(\val(x'))\supseteq \lambda_{x'}$, and
			\item for every edge $e=\set{x',y'}$ in $T_x$ with $x'<y'$ we have\\
			$(\val(x'),\val(y'))\in \EOne$ and $\el(\val(x'),\val(y'))\supseteq\lambda_e$.
		\end{enumerate}
		\item
		For every $y\in\Children(x)$, the number $g(c,y)$ is the number of mappings
		$\val\colon \set{x}\cup V(T_y) \to \VOne$ satisfying $\val(x)=v$ and
		\begin{enumerate}[(1)]
		\item for every $x'\in V(T_y)$ we have $\vl(\val(x'))\supseteq \lambda_{x'}$, and
		\item for every edge $e=\set{x',y'}$ in $T_x$ with $x'<y'$ and
		$x',y'\in \set{x}\cup V(T_y)$ we have\\
		$(\val(x'),\val(y'))\in \EOne$ and $\el(\val(x'),\val(y'))\supseteq\lambda_e$.
		\end{enumerate}
	\end{mea}
\end{lemma}

\noindent
For the particular case where $x =x_1$,
Lemmas~\ref{lemma:counting}(a) and~\ref{lemma:homomorphisms} yield that
\[
	\sum_{c\in C}\ n_c \cdot f_{\downarrow}(c,x_1)
	\quad \text{is} \quad
	\parbox[t]{10cm}{%
		the number of homomorphisms from $\QOne$ to $\DOne$,\\
		where $n_c$ is the number of nodes $v\in\VOne$ with $\col(v)=c$.%
	}
\]

It is straightforward to see that a lookup table that provides $O(1)$-time access to
$f_{\downarrow}(c,x)$ and $g(c,y)$ for all $c\in C$ and all $x,y\in\vars(\QOne)$
can be computed in time $O(|\QOne| {\cdot} |\ciD|)$.
Hence, also $\sum_{c\in C}\ n_c \cdot f_{\downarrow}(c,x_1)$ can be computed in time
$O(|\QOne| {\cdot} |\ciD|) = O(|Q| {\cdot}|\ciD|)$.

In the particular case where $\free(Q)=\vars(Q)$,
the number of homomorphisms from $\QOne$ to $\DOne$ is precisely the number $|\sem{\QOne}(\DOne)|$.
This solves the task $\counttask$ in case that $\free(\QOne)=\vars(\QOne)$.

To solve the task $\counttask$ in case that $\free(\QOne)\neq\vars(\QOne)$,
we additionally perform a bottom-up pass of the subtree $T'$ of $T$ induced by $\free(\QOne)$
(to see that $T'$ is indeed a tree, recall Proposition~\ref{prop:binaryfcacqs}).
For each leaf $x$ of $T'$ and every $c\in C$ we let
\[
	f'_{\downarrow}(c,x) \deff 1 \ \ \text{if} \
		f_\downarrow(c,x)\geq 1, \quad \text{and} \ \
		f'_{\downarrow}(c,x) \deff 0 \ \text{otherwise.}
\]
For every $c\in C$ and every node $y\neq x_1$ of $T'$, for $e\deff \set{\Parent(y),y}$ let
\[
	g'(c,y) \ \deff \ \
		\sum_{c'\in C} \ f'_{\downarrow}(c',y)  \cdot \ \hatnumSucc{\lambda_e}{c}{c'}.
\]
For every $c\in C$ and for every node $x$ of $T'$ that is not a leaf, let
\[
	f'_{\downarrow}(c,x)\ \deff \ \ f_1(c,x) \cdot\!\! \prod_{y\in\Children(x)} g'(c,y).
\]
By induction (bottom-up along $T'$) one can show that for every $c \in C$
and every $v \in \VOne$ with $\col(v)=c$, $f'_{\downarrow}(c,x_1)$ is the number of tuples
$(v_1,\ldots,v_k)\in \sem{\QOne}(\DOne)$ with $v_1=v$.
Thus,
\[
|\sem{\QOne}(\DOne)| \ \ = \ \ \sum_{c\in C} \ n_c \cdot f'_{\downarrow}(c,x_1).
\]
Furthermore, this number can be computed in time $O(|\QOne|{\cdot}|\ciD|)$.
In summary, we compute the number $|\sem{\QOne}(\DOne)|$ in time $O(|\QOne|{\cdot}|\ciD|)$. 

\section{Statement of our main result}\label{sec:MainTheorem}

From Sections~\ref{sec:mainresult} and~\ref{sec:eval} we directly obtain our following main result,
where $\ciD$ is the \emph{color database} produced for the input database $D$
during the indexing phase.
Note that $\ciD$ has size $O(|D|)$ in the worst-case;
but depending on the particular structure of $D$,
the size $|\ciD|$ might be substantially smaller than $|D|$ (even $O(1)$ is possible,
cf.\ Proposition~\ref{prop:colordbsize}).

\begin{theorem}\label{thm:MainTheorem}
	For every binary schema $\sigma$, the problem $\IndexingProblemOurs$ can be solved with 
	indexing time $\fcr(D)+ O(|D|)$, such that afterwards, 
	every Boolean acyclic query $Q$ can be answered in time $O(|Q|{\cdot}|\ciD|)$.
	Furthermore, for every query $Q\in\fcACQ$  we can enumerate the tuples in $\sem{Q}(D)$ with 
	delay $O(|\free(Q)|)$ after preprocessing time $O(|Q|{\cdot}|\ciD|)$, 
	and we can compute the number $|\sem{Q}(D)|$ of result tuples in time $O(|Q|{\cdot}|\ciD|)$.
\end{theorem} 

\section{Final Remarks}\label{sec:conclusion}
Our main result provides a solution for the problem $\IndexingProblemOurs$ for any
\emph{binary} schema $\sigma$.
Our approach can be summarized as follows.
Upon input of a database $D$ of schema $\sigma$, during the indexing phase we construct a
vertex-labeled and edge-labeled directed graph $\GOne$ that represents $D$.
We then use a suitable version of the well-known color refinement algorithm to construct a
coarsest stable coloring $\col\colon\Adom(D)\to C$ (for a suitable set $C$ of colors)
that refines the vertex-labels of $\GOne$.
Based on this, we construct the data structure $\DSD$ that we call \emph{color-index} and
that we view as the result of the indexing phase.
The main ingredient of $\DSD$ is the \emph{color database} $\ciD$.
This is a database of a newly constructed schema $\cisigma$; its active domain is the set $C$.

During the evaluation phase, we can access the color-index~$\DSD$.
When receiving an arbitrary $k$-ary input query $Q\in \fcACQ$, we use $\DSD$
(and, in particular, $\ciD$) to
\begin{itemize}
	\item
	answer $Q$, in case that $k=0$, in time $O(|Q|{\cdot}|\ciD|)$,
	\item
	compute $|\sem{Q}(D)|$ in time $O(|Q|{\cdot}|\ciD|)$, and
	\item
	enumerate $\sem{Q}(D)$ with delay $O(k)$ after $O(|Q|{\cdot}|\ciD|)$ preprocessing.
\end{itemize}

\noindent
It was known previously that each of these query evaluation tasks can be solved
(without access to our color-index $\DSD$) in time $O(|Q|{\cdot}|D|)$. 
Our result provides an improvement from $|D|$ to $|\ciD|$.
In the worst-case, $|\ciD|=O(|D|)$; but depending on the particular structure of $D$,
the color database $\ciD$ may be much smaller than $D$ (see Proposition~\ref{prop:colordbsize}).

\paragraph{Future work} 
Obvious questions (we currently work on) are:

\emph{Can our approach be generalized from fc-ACQs to queries of free-connex tree-width $\leq k$,
for any fixed $k \geq 1$?}
When sticking with binary schemas, we believe that this can be achieved by using a variant of
the $k$-dimensional Weisfeiler-Leman algorithm (cf.~\cite{CFI-paper,grohe_color_2021}).

\emph{Can our approach be generalized from binary schemas to arbitrary schemas?}
Due to subtleties concerning the notion of fc-ACQs, a straightforward representation of databases
over an arbitrary schema $\sigma$ by databases of a suitably chosen binary schema $\sigma'$
apparently does not help to easily lift our results from binary to arbitrary schemas
(see Appendix~\ref{appendix:BinaryToArbitrary}).
To solve the case for arbitrary schemas, a generalization of the color refinement algorithm from
graphs to hypergraphs would be desirable.
Böker~\cite{Boeker2019} provided a color refinement algorithm that seems promising for handling
queries $Q$ that are free-connex \emph{Berge-acyclic}.
But a new color refinement algorithm needs to be developed to handle also
free-connex acyclic queries that are not free-connex Berge-acyclic.

\emph{Can our approach be lifted to a dynamic scenario?} 
The \emph{dynamic Yannakakis} approach of Idris, Ugarte, and Vansummeren~%
\cite{DynamicYannakakis2017} lifts Theorem~\ref{thm:BGS-enum} to a dynamic scenario,
where a fixed query $Q$ shall be evaluated against a database $D$ that is frequently updated.
Note that our color-index $\DSD$ is not built for a fixed query,
but for supporting \emph{all} queries.
When receiving an update for the database $D$,
we would have to modify the color-index $\DSD$ accordingly.
Already a small change of the color of a single vertex or its local neighborhood in $\GOne$ might
impose drastic changes in the result of the color refinement algorithm and might change
the size of $\ciD$ from $O(1)$ to $\Theta(|D|)$, or vice versa.
(For example, consider the database $D^n$ from the proof of Proposition~\ref{prop:colordbsize}.
Inserting the tuple $(1,1)$ into the $R$-relation leads to 
a coarsest stable coloring with approximately $n/2$ colors;
and subsequently deleting $(1,1)$ from the $R$-relation again reduces the number of colors to 1.)
Nevertheless, it is an interesting question if $\DSD$ can be updated in time proportional to
the actual difference between the two versions of $\DSD$ before and after the update.

\emph{What is the potential practical impact of our approach?}
To further reduce the size of $\ciD$,
one may decide to let the color refinement procedure only run for a fixed number $r$ of rounds.
The resulting coarser variant of our color-index can still be used
for evaluating fc-ACQs of radius up to $r$.
We believe that this variant of the color-index deserves a detailed experimental evaluation.
 
\clearpage

\bibliography{literature}

\begin{thebibliography}{10}

\bibitem{AHV-Book}
Serge Abiteboul, Richard Hull, and Victor Vianu.
\newblock {\em Foundations of Databases}.
\newblock Addison-Wesley, 1995.

\bibitem{AnglesABHRV17}
Renzo Angles, Marcelo Arenas, Pablo Barcel{\'{o}}, Aidan Hogan, Juan~L.
  Reutter, and Domagoj Vrgoc.
\newblock Foundations of modern query languages for graph databases.
\newblock {\em {ACM} Comput. Surv.}, 50(5):68:1--68:40, 2017.

\bibitem{ArroyueloHNRRS21}
Diego Arroyuelo, Aidan Hogan, Gonzalo Navarro, Juan~L. Reutter, Javiel
  Rojas{-}Ledesma, and Adri{\'{a}}n Soto.
\newblock Worst-case optimal graph joins in almost no space.
\newblock In {\em SIGMOD}, pages 102--114. {ACM}, 2021.

\bibitem{Arvind2017}
V.~Arvind, Johannes Köbler, Gaurav Rattan, and Oleg Verbitsky.
\newblock Graph {{Isomorphism}}, {{Color Refinement}}, and {{Compactness}}.
\newblock {\em Computational Complexity}, 26(3):627--685, September 2017.

\bibitem{Bagan_PhD}
Guillaume Bagan.
\newblock {\em Algorithmes et complexit{\'{e}} des probl{\`{e}}mes
  d'{\'{e}}num{\'{e}}ration pour l'{\'{e}}valuation de requ{\^{e}}tes logiques.
  (Algorithms and complexity of enumeration problems for the evaluation of
  logical queries)}.
\newblock PhD thesis, University of Caen Normandy, France, 2009.

\bibitem{Bagan.2007}
Guillaume Bagan, Arnaud Durand, and Etienne Grandjean.
\newblock On acyclic conjunctive queries and constant delay enumeration.
\newblock In {\em Proceedings of the 16th Annual Conference of the EACSL,
  CSL'07, Lausanne, Switzerland, September 11--15, 2007}, pages 208--222, 2007.

\bibitem{DBLP:journals/jacm/BeeriFMY83}
Catriel Beeri, Ronald Fagin, David Maier, and Mihalis Yannakakis.
\newblock On the desirability of acyclic database schemes.
\newblock {\em J. {ACM}}, 30(3):479--513, 1983.

\bibitem{BBG-ColorRefinement}
Christoph Berkholz, Paul~S. Bonsma, and Martin Grohe.
\newblock Tight lower and upper bounds for the complexity of canonical colour
  refinement.
\newblock {\em Theory Comput. Syst.}, 60(4):581--614, 2017.

\bibitem{BGS-tutorial}
Christoph Berkholz, Fabian Gerhardt, and Nicole Schweikardt.
\newblock Constant delay enumeration for conjunctive queries: a tutorial.
\newblock {\em {ACM} {SIGLOG} News}, 7(1):4--33, 2020.

\bibitem{DBLP:journals/siamcomp/BernsteinG81}
Philip~A. Bernstein and Nathan Goodman.
\newblock Power of natural semijoins.
\newblock {\em {SIAM} J. Comput.}, 10(4):751--771, 1981.

\bibitem{Bollen2023}
Jeroen Bollen, Jasper Steegmans, Jan Van Den~Bussche, and Stijn Vansummeren.
\newblock {L}earning {G}raph {N}eural {N}etworks using {E}xact {C}ompression.
\newblock In {\em Proceedings of the 6th Joint Workshop on Graph Data
  Management Experiences \& Systems (GRADES) and Network Data Analytics (NDA)},
  GRADES \& NDA '23, New York, NY, USA, 2023. Association for Computing
  Machinery.

\bibitem{BraultBaron_PhD}
Johann Brault{-}Baron.
\newblock {\em De la pertinence de l'{\'{e}}num{\'{e}}ration : complexit{\'{e}}
  en logiques propositionnelle et du premier ordre. (The relevance of the list:
  propositional logic and complexity of the first order)}.
\newblock PhD thesis, University of Caen Normandy, France, 2013.

\bibitem{BraultBaron16}
Johann Brault-Baron.
\newblock {H}ypergraph {A}cyclicity {R}evisited.
\newblock {\em {ACM} Comput. Surv.}, 49(3):54:1--54:26, 2016.

\bibitem{Boeker2019}
Jan Böker.
\newblock {C}olor {R}efinement, {H}omomorphisms, and {H}ypergraphs.
\newblock In Ignas Sau and Dimitrios~M. Thilikos, editors, {\em Graph-Theoretic
  Concepts in Computer Science}, volume 11789 of {\em Lecture Notes in Computer
  Science}, pages 338--350. Springer, Cham, 2019.

\bibitem{CFI-paper}
Jin{-}yi Cai, Martin F{\"{u}}rer, and Neil Immerman.
\newblock An optimal lower bound on the number of variables for graph
  identification.
\newblock {\em Combinatorica}, 12(4):389--410, 1992.

\bibitem{Dvorak2010}
Zdeněk Dvořák.
\newblock {O}n {R}ecognizing {G}raphs by {N}umbers of {H}omomorphisms.
\newblock {\em Journal of Graph Theory}, 64(4):330--342, 2010.

\bibitem{DBLP:journals/jcss/GottlobLS02}
Georg Gottlob, Nicola Leone, and Francesco Scarcello.
\newblock Hypertree decompositions and tractable queries.
\newblock {\em J. Comput. Syst. Sci.}, 64(3):579--627, 2002.

\bibitem{Grohe2020a}
Martin Grohe.
\newblock {W}ord2vec, {N}ode2vec, {G}raph2vec, {X}2vec: {T}owards a {T}heory of
  {V}ector {E}mbeddings of {S}tructured {D}ata.
\newblock In {\em Proceedings of the 39th ACM SIGMOD-SIGACT-SIGAI Symposium on
  Principles of Database Systems}, PODS'20, pages 1--16. ACM, 2020.

\bibitem{grohe_color_2021}
Martin Grohe, Kristian Kersting, Martin Mladenov, and Pascal Schweitzer.
\newblock Color {Refinement} and {Its} {Applications}.
\newblock In Guy Van~den Broeck, Kristian Kersting, Sriraam Natarajan, and
  David Poole, editors, {\em An {Introduction} to {Lifted} {Probabilistic}
  {Inference}}. The MIT Press, 2021.

\bibitem{grohe_dimension_2014}
Martin Grohe, Kristian Kersting, Martin Mladenov, and Erkal Selman.
\newblock Dimension {Reduction} via {Colour} {Refinement}.
\newblock In Andreas~S. Schulz and Dorothea Wagner, editors, {\em Algorithms -
  {ESA} 2014}, Lecture {Notes} in {Computer} {Science}, pages 505--516, Berlin,
  Heidelberg, 2014. Springer.

\bibitem{DynamicYannakakis2017}
Muhammad Idris, Mart{\'{\i}}n Ugarte, and Stijn Vansummeren.
\newblock The {D}ynamic {Y}annakakis algorithm: Compact and efficient query
  processing under updates.
\newblock In {\em Proc.\ 2017 {ACM} International Conference on Management of
  Data ({SIGMOD} Conference 2017), Chicago, IL, USA, May 14--19, 2017}, pages
  1259--1274, 2017.

\bibitem{DBLP:journals/pvldb/IdrisUVVL18}
Muhammad Idris, Mart{\'{\i}}n Ugarte, Stijn Vansummeren, Hannes Voigt, and
  Wolfgang Lehner.
\newblock Conjunctive queries with inequalities under updates.
\newblock {\em {PVLDB}}, 11(7):733--745, 2018.

\bibitem{DBLP:journals/sigmod/IdrisUVVL19}
Muhammad Idris, Mart{\'{\i}}n Ugarte, Stijn Vansummeren, Hannes Voigt, and
  Wolfgang Lehner.
\newblock Efficient query processing for dynamically changing datasets.
\newblock {\em {SIGMOD} Record}, 48(1):33--40, 2019.

\bibitem{Kayali2022}
Moe Kayali and Dan Suciu.
\newblock Quasi-{{Stable Coloring}} for {{Graph Compression}}: {{Approximating
  Max-Flow}}, {{Linear Programs}}, and {{Centrality}}.
\newblock In {\em Proceedings of the {{VLDB Endowment}}}, volume~16, pages
  803--815, December 2022.

\bibitem{Kiefer2020}
Sandra Kiefer.
\newblock The {{Weisfeiler-Leman Algorithm}}: {{An Exploration}} of {{Its
  Power}}.
\newblock {\em ACM SIGLOG News}, 7(3):5--27, November 2020.

\bibitem{Kiefer2021}
Sandra Kiefer, Pascal Schweitzer, and Erkal Selman.
\newblock Graphs {{Identified}} by {{Logics}} with {{Counting}}.
\newblock {\em ACM Transactions on Computational Logic}, 23(1):1:1--1:31,
  October 2021.

\bibitem{ngo2018worst}
Hung~Q Ngo, Ely Porat, Christopher R{\'e}, and Atri Rudra.
\newblock Worst-case optimal join algorithms.
\newblock {\em Journal of the ACM (JACM)}, 65(3):1--40, 2018.

\bibitem{DBLP:journals/tods/OlteanuZ15}
Dan Olteanu and Jakub Z{\'{a}}vodn{\'{y}}.
\newblock Size bounds for factorised representations of query results.
\newblock {\em {ACM} Trans. Database Syst.}, 40(1):2:1--2:44, 2015.

\bibitem{Picalausa2014}
François Picalausa, George H.~L. Fletcher, Jan Hidders, and Stijn Vansummeren.
\newblock Principles of {Guarded Structural Indexing}.
\newblock In Nicole Schweikardt, Vassilis Christophides, and Vincent Leroy,
  editors, {\em Proc. 17th {International Conference} on {Database Theory}
  ({ICDT}), {Athens}, {Greece}, {March} 24-28, 2014}, pages 245--256.
  OpenProceedings.org, 2014.

\bibitem{Picalausa2012}
François Picalausa, Yongming Luo, George H.~L. Fletcher, Jan Hidders, and
  Stijn Vansummeren.
\newblock A {Structural Approach} to {Indexing Triples}.
\newblock In Elena Simperl, Philipp Cimiano, Axel Polleres, Oscar Corcho, and
  Valentina Presutti, editors, {\em The {Semantic Web}: {Research} and
  {Applications}}, Lecture {{Notes}} in {{Computer Science}}, pages 406--421,
  Berlin, Heidelberg, 2012. Springer.

\bibitem{ramakrishnan2003database}
Raghu Ramakrishnan and Johannes Gehrke.
\newblock {\em Database management systems}, volume~3.
\newblock McGraw-Hill New York, 2003.

\bibitem{RSS26}
Cristian Riveros, Benjamin Scheidt, and Nicole Schweikardt.
\newblock Structural indexing of relational databases for the evaluation of
  free-connex acyclic conjunctive queries, 2026.

\bibitem{Tran2013}
Thanh Tran, Gunter Ladwig, and Sebastian Rudolph.
\newblock Managing {{Structured}} and {{Semistructured RDF Data Using Structure
  Indexes}}.
\newblock {\em IEEE Transactions on Knowledge and Data Engineering},
  25(9):2076--2089, September 2013.

\bibitem{DBLP:conf/icdt/Veldhuizen14}
Todd~L. Veldhuizen.
\newblock Triejoin: {A} simple, worst-case optimal join algorithm.
\newblock In Nicole Schweikardt, Vassilis Christophides, and Vincent Leroy,
  editors, {\em Proc. 17th International Conference on Database Theory (ICDT),
  Athens, Greece, March 24-28, 2014}, pages 96--106. OpenProceedings.org, 2014.

\bibitem{Yannakakis1981}
Mihalis Yannakakis.
\newblock Algorithms for acyclic database schemes.
\newblock In {\em Very Large Data Bases, 7th International Conference,
  September 9-11, 1981, Cannes, France, Proceedings}, pages 82--94, 1981.

\end{thebibliography}

\newpage

\appendix

\section*{APPENDIX}

\section{Details on acyclic and free-connex acyclic conjunctive queries}\label{appendix:fcACQ}

\subsection{Proof of Proposition~\ref{prop:binaryfcacqs}}
Let $\sigma$ be a binary schema, i.e., every $R\in\sigma$ has arity $\ar(R)\leq 2$.
In the following, if $G$ is a (hyper)graph, then we denote by $G - e$ the resulting (hyper)graph
after removing the (hyper)edge $e \in E(G)$ from $G$.
If a vertex becomes isolated due to this procedure, we remove it as well.
Further, we denote by $G + e$ the (hyper)graph resulting from adding the (hyper)edge $e$ to $G$.
First we show that we can assume that $Q$ contains no loops
and only uses predicates of arity exactly 2.
\begin{claim}
	$Q$ is free-connex acyclic iff $Q'$ is free-connex acyclic,
	where $\Atoms(Q') \isdef \set{ R(x,y) \mid R(x,y) \in \Atoms(Q),\ x {\neq} y }$
	and $\free(Q') \isdef \free(Q) \cap \Vars(Q')$.
\end{claim}
\begin{proof}
	\noindent ($\Rightarrow$)\; Let $Q$ be free-connex acyclic.
	Then $H(Q) + \free(Q)$ has a join-tree $T$.
	We translate $T$ into a join-tree $T'$ for $H(Q')+\free(Q')$ as follows.
	First, we replace the node $\free(Q)$ in $T$ with $\free(Q')$.
	Then, for every hyperedge $e = \set{ x }$,
	we look for a neighbor $f$ of $e$ in $T$ such that $x \in f$.
	If there is none, we let $f$ be an arbitrary neighbor.
	Then we connect all other neighbors of $e$ in $T$ with $f$ and remove $e$ from $T$.
	The resulting graph is still a tree and the set of nodes containing $x$ still
	induces a connected subtree, or it is empty --- but in that case also $x \not\in \Vars(Q')$.
	I.e., now the set of nodes in $T'$ is precisely the set of hyperedges of $H(Q')+\free(Q')$,
	and for every $x \in \Vars(Q')$,
	the set $\set{ t \in V(T') \mid x \in t }$ still induces a connected subtree.
	Hence, $T'$ is a join-tree for $H(Q') + \free(Q')$.

	Since $Q$ is free-connex acyclic, there also exists a join-tree $T$ for $H(Q)$.
	With the same reasoning as above, $T$ can be transformed into a join-tree for $H(Q')$.
	In summary, we obtain that $Q'$ is free-connex acyclic.
	\medskip

	\noindent ($\Leftarrow$)\;
	Let $Q'$ be free-connex acyclic.
	Then $H(Q') + \free(Q')$ has a join-tree $T'$.
	We can extend $T'$ to a join-tree $T$ for $H(Q) + \free(Q)$ as follows.
	First, we replace the node $\free(Q')$ in $T'$ with $\free(Q)$.
	For every hyperedge $e$ in $H(Q)$ of the form $e = \set{ x }$
	we add $e$ as a new node to $T'$ and insert the edge $\set{ e, f }$ into $T'$ as follows:
	If $x \not\in \Vars(Q')$ we let $f$ be an arbitrary node in $V(T')$.
	Otherwise, let $f \in V(T')$ be a node such that $x \in f$.
	It is easy to verify that the resulting tree $T$ is a join-tree for $H(Q) + \free(Q)$.
	Since $Q'$ is free-connex acyclic, there also exists a join-tree $T$ for $H(Q')$.
	With the same reasoning as above, $T'$ can be transformed into a join-tree for $H(Q)$.
	In summary, we obtain that $Q$ is free-connex acyclic.
\end{proof}

Due to the above claim, we can assume w.l.o.g.\
that $Q$ contains no self-loops and only uses predicates of arity 2.
Notice that this means that $H(Q) = G(Q)$.
It is well-known that for undirected graphs the notion of
$\alpha$-acyclicity and the standard notion of acyclicity for graphs coincide.
I.e., a graph is $\alpha$-acyclic iff it is acyclic (see~\cite{BraultBaron16} for an overview).
Hence, in our setting, $H(Q)=G(Q)$ is $\alpha$-acyclic,
if and only if it is acyclic, i.e., a forest.
This immediately yields the first statement of Proposition~\ref{prop:binaryfcacqs}.
The second statement of Proposition~\ref{prop:binaryfcacqs} is obtained by the following claim.

\begin{claim}\label{appendix:claim:1forprop}
	$Q$ is free-connex acyclic iff $G(Q)$ is acyclic and the following statement is true:
	(**) For every connected component $C$ of $G(Q)$,
	the subgraph of $C$ induced by the set $\free(Q) \cap V(C)$ is connected or empty.
\end{claim}

\begin{proof}
\noindent $(\Rightarrow)$\;
By assumption, $Q$ is free-connex acyclic.
I.e., $H(Q)$ is $\alpha$-acyclic and $H(Q)+\free(Q)$ is $\alpha$-acyclic.
Since $H(Q) = G(Q)$, and since on graphs $\alpha$-acyclicity coincides with acyclicity,
this means that $G(Q)$ is acyclic, i.e., it is a forest.

Since $H(Q)+\free(Q)$ is $\alpha$-acyclic, there exists a join-tree $T$ for $H(Q)+\free(Q)$.
To prove (**) we proceed by induction on the number of nodes in the join-tree of $H(Q) + \free(Q)$.

If $T$ has at most two nodes, $H(Q)$ consists of a single hyperedge,
and this is of the form $\set{ x,y }$ with $x \neq y$.
Therefore, $G(Q)$ consists of a single edge $\set{ x,y }$, and in this case (**) trivially holds.

In the inductive case, let $T$ be a join tree of $H(Q) + \free(Q)$.
Consider some leaf $e$ of $T$ (i.e, $e$ only has one neighbor in $T$)
of the form $e = \set{ x,y }$ and its parent $p$.
Then, $T - e$ is a join-tree for $H(Q') + \free(Q)$ where
$\Atoms(Q') \isdef \set{ R(z_1, z_2) \in \Atoms(Q) \mid \set{ z_1, z_2 } \neq \set{ x, y } }$
and $\free(Q')$ is the set of all variables in $\free(Q)$ that occur in an atom of $Q'$.
It is easy to see that this implies the existence of a join-tree $T'$ for $H(Q') + \free(Q')$.
Since, furthermore, $G(Q')$ is acyclic, the query $Q'$ is free-connex acyclic,
and by induction hypothesis, (**) is true for $G(Q')$.
Note that $G(Q')=G(Q)-e$.
Obviously, $x$ and $y$ are in the same connected component $C$ in $G(Q)$.
We have to consider the relationship between $x$, $y$ and the rest of $C$.

Since $e$ is a leaf in $T$, $x,y \in p$ would imply that $p = \free(Q)$.
Using that we know that
\begin{enumerate}[(i)]
	\item either $y$ only has $x$ as a neighbor in $G(Q)$ or vice-versa, or
	\item $x$ and $y$ are both free variables and both have other neighbors
	than $y$ and $x$, respectively.
\end{enumerate}

\emph{Case (i)}:\; Assume w.l.o.g.\ that $x$ is the only vertex adjacent to $y$ in $G(Q)$.
Then $C' \isdef C - \set{ x,y }$ is a connected component in $G(Q')$, $y \not\in \Vars(Q')$,
i.e., in particular also $y \not\in \free(Q')$.
By induction hypothesis, $\free(Q') \cap V(C')$ induces a connected subgraph on $C'$.
If $y \not\in \free(Q)$ it follows trivially,
that $\free(Q) \cap V(C)$ induces a connected subgraph on $C$.
If $y \in \free(Q)$ and $\free(Q) \cap V(C) = \set{ y }$, this is also trivial.
Otherwise, there exists a $z \in \free(Q) \cap V(C)$ different from $y$.
Since $x,y,z$ are all in the same connected component $C$,
but $x$ is the only vertex adjacent to $y$,
there must be another vertex $w$ that is adjacent to $x$.
Therefore, $x$ is part of another node $\set{ x, w }$ somewhere in $T$,
which by definition means that $x$ must be in $p$.
Since $y, z \in \free(Q)$, $y$ must also be in $p$.
Therefore, $p = \free(Q)$ and that means $x \in \free(Q)$.
Hence, $\free(Q) \cap V(C)$ forms a connected component on $C$ in this case as well.

\emph{Case (ii)}:\; Assume that both variables $x$ and $y$ are free
and both have another vertex adjacent to them in $G(Q)$.
Then $x,y \in \free(Q')$ (i.e., $\free(Q) = \free(Q')$) and removing the edge $\set{ x,y }$ splits
$C$ into two connected components $C_x$, $C_y$.
By induction hypothesis, $\free(Q') \cap V(C_x)$ and $\free(Q') \cap V(C_y)$ induce connected
subgraphs on $C_x$ and $C_y$, respectively, and $x$ (and $y$, resp.) are part of them.
Thus, adding the edge $\set{ x,y }$ establishes a connection between them,
so $\free(Q) \cap V(C)$ also induces a connected subtree on $C$.

Every other connected component $D \neq C$ in $G(Q)$ is also one in $G(Q')$.
Thus, (**) is true for $G(Q)$.
\medskip

\noindent $(\Leftarrow)$\;
By assumption, $G(Q)$ is acyclic and (**) is satisfied.
Since $G(Q)=H(Q)$ and $\alpha$-acyclicity coincides with acyclicity on graphs, $Q$ is acyclic.
It remains to show that $H(Q)+\free(Q)$ has a join-tree.

We proceed by induction over the number of vertices in $G(Q)$.
Recall that $Q$ only uses binary relation symbols and that it has no self-loops.
Thus, in the base case, $G(Q)$ consists of two vertices connected by an edge.
Since $H(Q) = G(Q)$, it is easy to see that $H(Q) + \free(Q)$ has a join-tree.

In the inductive case, let $C$ be a connected component of $G(Q)$
that contains a quantified variable, i.e., $\free(Q) \cap V(C) \neq V(C)$.
If no such component exists, $H(Q) + \free(Q)$ has a trivial join-tree by letting
all hyperedges of $H(Q)$ be children of $\free(Q)$ in $T$.

Because of (**) there must be a leaf $x$ in $C$ such that $x \not\in \free(Q)$.
Let $y$ be the parent of $x$, i.e., let $\set{ x,y }$ be an (or rather, the only)
edge in $G(Q)$ that contains $x$.
Let $H' \isdef H(Q) - \set{ x,y }$.
Then $H' = H(Q')$ for the query $Q'$ where
$\Atoms(Q') = \set{ R(z_1, z_2)\in\Atoms(Q) \mid x \not\in \set{ z_1, z_2 } }$
and $\free(Q') = \free(Q) \setminus \set{ x }$.

Let $T'$ be the join-tree that exists by induction hypothesis for
$H(Q') + \free(Q')$, i.e., for $H' + \free(Q) \setminus \set{ x }$.

If $x \not\in \free(Q)$, we can extend $T'$ to a join-tree $T$ for $H(Q)$ by
adding the edge $\set{ x,y }$ as a child node to some node that contains $y$ in $T'$.
We can similarly handle the case $\free(Q) = \set{ x }$
by also inserting $\set{ x }$ as a child of $\set{ x,y }$.

If $x \in \free(Q)$ but $\free(Q) \varsupsetneqq \set{ x }$, it holds that $y \in \free(Q)$,
by the assumption that $\free(Q) \cap V(C)$ induces a connected subgraph.
In this case we can add $\set{ x,y }$ as a child of $\free(Q)$.
\end{proof}

\subsection{Proof of Theorem~\ref{thm:BGS-enum}}
We provide further definitions that are taken from~\cite{BGS-tutorial}.
We use these definitions for the proof of Theorem~\ref{thm:BGS-enum} below.

A \emph{tree decomposition} (TD, for short) of a CQ $Q$ and its hypergraph $H(Q)$
is a tuple $\TD=(\Tree,\Bag)$, such that:
\begin{enumerate}[(1)]
	\item
	$\Tree=(\Nodes(\Tree),\Edges(\Tree))$ is a finite undirected tree, and
	\item
	$\Bag$ is a mapping that associates with every node $\treenode\in\Nodes(\Tree)$
	a set $\Bag(\treenode)\subseteq \Vars(\query)$ such that
	\begin{enumerate}
		\item
		for each atom $\qatom\in\Atoms(\query)$ there exists $\treenode\in\Nodes(\Tree)$
		such that $\Vars(\qatom)\subseteq\Bag(\treenode)$,
		\item\label{item:pathcondition:treedecomp}
		for each variable $v\in\Vars(\query)$ the set
		$\Bag^{-1}(v)\deff\setc{\treenode\in\Nodes(\Tree)}{v\in\Bag(\treenode)}$
		induces a connected subtree of $\Tree$.
	\end{enumerate}
\end{enumerate}
The \emph{width} of $\TD=(\Tree,\Bag)$ is defined as
$\Width(\TD)=\max_{\treenode\in\Nodes(\Tree)}|\Bag(t)| -1$.

A \emph{generalized hypertree decomposition} (GHD, for short) of a CQ $\query$
and its hypergraph $H(Q)$ is a tuple $\HD=(\Tree,\Bag,\Cover)$ that consists of
a tree decomposition $(\Tree,\Bag)$ of $\query$ and a mapping $\Cover$ that associates with
every node $\treenode\in\Nodes(\Tree)$ a set $\Cover(\treenode)\subseteq\Atoms(\query)$ such that
$\Bag(\treenode)\subseteq\bigcup_{\qatom\in\Cover(\treenode)}\Vars(\qatom)$.
The sets $\Bag(\treenode)$ and $\Cover(\treenode)$ are called the \emph{bag} and the \emph{cover}
associated with node $\treenode\in\Nodes(\Tree)$.
The \emph{width} of a GHD $\HD$ is defined as the maximum number of atoms in a $\Cover$-label
of a node of $\Tree$, i.e., $\Width(\HD)=\max_{\treenode\in\Nodes(\Tree)}|\Cover(t)|$.
The \emph{generalized hypertree width} of a CQ $\query$, denoted $\GHW(\query)$,
is defined as the minimum width over all its generalized hypertree decompositions.

A tree decomposition $\TD=(\Tree,\Bag)$ of a CQ $\query$ is \emph{free-connex} if there is a set
$\freetreenodes \subseteq \Nodes(\Tree)$ that induces a connected subtree of $\Tree$
and that satisfies the condition
$\free(\query) = \bigcup_{\treenode\in\freetreenodes} \Bag(\treenode)$.
Such a set $\freetreenodes$ is called a \emph{witness} for the free-connexness of $\TD$.
A GHD is free-connex if its tree decomposition is free-connex.
The \emph{free-connex generalized hypertree width} of a CQ $\query$, denoted $\fcGHW(\query)$,
is defined as the minimum width over all its free-connex generalized hypertree decompositions.

It is known (\cite{DBLP:journals/jcss/GottlobLS02,Bagan.2007,BraultBaron_PhD};
see also~\cite{BGS-tutorial} for an overview as well as proof details)
that the following is true for every schema $\sigma$ and every CQ $Q$ of schema $\sigma$:
\begin{enumerate}[(I)]
	\item\label{item:acqchara}
	$Q$ is acyclic iff $\GHW(Q)=1$.
	\item\label{item:fcacqchara}
	$Q$ is free-connex acyclic iff $\fcGHW(Q)=1$.
\end{enumerate}

A GHD $\HD=(\Tree,\Bag,\Cover)$ of a CQ $\query$ is called \emph{complete} if,
for each atom $\qatom\in\Atoms(\query)$ there exists a node $\treenode\in\Nodes(\Tree)$
such that $\Vars(\qatom)\subseteq\Bag(\treenode)$ and $\qatom\in\Cover(\treenode)$.

The following has been shown in~\cite{BGS-tutorial} (see Theorem~4.2 in~\cite{BGS-tutorial}):
\begin{theorem}[\cite{BGS-tutorial}]\label{thm:BGS-refined}
	For every schema $\sigma$ there is an algorithm which receives as input a query $Q\in\fcACQ$,
	a complete free-connex width~1 GHD $\HD=(\Tree,\Bag,\Cover)$ of $\query$
	along with a witness $\freetreenodes \subseteq\Nodes(\Tree)$,
	and a $\sigma$-db $D$ and computes within preprocessing time $O(|\Nodes(\Tree)|{\cdot}|D|)$
	a data structure that allows to enumerate $\sem{\query}(D)$ with delay $O(|\freetreenodes|)$.
\end{theorem}

It is known that for every schema $\sigma$ there is an algorithm that receives as input
a query $Q\in\fcACQ$ and computes in time $O(|Q|)$ a free-connex width 1 GHD of $Q$
(\cite{Bagan_PhD}; see also Section~5.1 in~\cite{BGS-tutorial}).
When applying to this GHD the construction used in the proof of~\cite[Lemma~3.3]{BGS-tutorial},
and afterwards performing the completion construction from~\cite[Remark~3.1]{BGS-tutorial},
one can compute in time $O(|Q|)$ a GHD $\HD=(\Tree,\Bag,\Cover)$ of $\query$ along with a witness
$\freetreenodes \subseteq\Nodes(\Tree)$ that satisfy the assumptions of
Theorem~\ref{thm:BGS-refined} and for which, additionally,
the following is true: $|U|\leq |\free(Q)|$ and $|\Nodes(\Tree)|\in O(|Q|)$.
In summary, this provides a proof of Theorem~\ref{thm:BGS-enum}.

\subsection{An example concerning a non-binary schema}\label{appendix:BinaryToArbitrary}
At a first glance one may be tempted to believe that it should be straightforward to generalize
our main result from binary to arbitrary schemas by representing databases of
a non-binary schema $\sigma$ as databases of a suitably chosen binary schema $\sigma'$.
However, the notion of fc-ACQs is quite subtle, and an fc-ACQ of schema $\sigma$ might translate
into a CQ of schema $\sigma'$ that is not an fc-ACQ, hence prohibiting to straightforwardly
reduce the case of general schemas to our results for fc-ACQs over binary schemas.
Here is a concrete example.

Let us consider a schema consisting of a single, ternary relation symbol $R$.
A straightforward way to represent a database $D$ of this schema by
an edge-labeled bipartite digraph $D'$ is as follows.
$D'$ has 3 edge-relations called $E_1, E_2, E_3$.
Every element in the active domain of $D$ is a node of $D'$.
Furthermore, every tuple $t = (a_1,a_2,a_3)$ in the $R$-relation of $D$ serves as a node of $D'$,
and we insert into $D'$ an $E_1$-edge $(t, a_1)$, an $E_2$-edge $(t,a_2)$,
and an $E_3$-edge $(t,a_3)$.
Now, a CQ $Q$ posed against $D$ translates into a CQ $Q'$ posed against $D'$ as follows:
$Q'$ has the same head as $Q$.
For each atom $R(x,y,z)$ in the body of $Q$ we introduce a new variable $u$ and insert
into the body of $Q'$ the atoms $E_1(u,x), E_2(u,y), E_3(u,z)$.

We want to build an index structure on $D$ that supports constant-delay enumeration and counting
for all fc-ACQs $Q$ posed against $D$.
The problem is that an fc-ACQ $Q$ against $D$ does \emph{not} necessarily translate
into an \emph{fc-ACQ} $Q'$ against $D'$ --- and therefore,
lifting our approach from binary to arbitrary schemas is not so easy.
Here is a specific example of such a query $Q$:
\[
	\Ans(x,y,z) \ \leftarrow \ R(x,y,z), \ R(x,x,y), \ R(y,y,z), \ R(z,z,x).
\]
This query is an fc-ACQ (as a witness, take the join-tree whose root is labeled with $R(x,y,z)$
and has 3 children labeled with the remaining atoms in the body of the query).
But the associated query $Q'$ is
\[
\Ans(x,y,z) \ \leftarrow
	\begin{array}[t]{l}
		E_1(u_1,x), E_2(u_1,y), E_3(u_1,z), \\
		E_1(u_2,x), E_2(u_2,x), E_3(u_2,y), \\
		E_1(u_3,y), E_2(u_3,y), E_3(u_3,z), \\
		E_1(u_4,z), E_2(u_4,z), E_3(u_4,x).
	\end{array}
\]
Note that the Gaifman-graph of $Q'$ is not acyclic (it contains the cycle $x-u_2-y-u_3-z-u_4-x$).
Therefore, by Proposition~\ref{prop:binaryfcacqs}, $Q'$ is not an fc-ACQ\@.
This simple example illustrates that the straightforward encoding of the database $D$
as a database over a binary schema won't help to easily lift the main result of
this paper from binary to arbitrary schemas.

\section{Proof of Theorem~\ref{thm:ColorRefinement}}\label{appendix:proof_of_runtime}

\begin{proof}[Proof of Theorem~\ref{thm:ColorRefinement}]
	We slightly adapt the proof of Theorem~12 in~\cite{BBG-ColorRefinement}
	since we also have to consider vertex colors.
	Nothing interesting happens when we do that, though.

	Let $G' = (\VOne \cup V', E', \alpha)$ be the vertex-colored digraph where
	$V' \isdef \set{ v_{(u,w)} \mid (u,w) \in \EOne }$,
	$E' \isdef \set{ (u, v_{u,w}) \mid (u,w) \in \EOne }
		\cup \set{ (v_{u,w}, w) \mid (u,w) \in \EOne }$
	and $\alpha(u) \isdef \vl(u)$ for all $u \in \VOne$
	and $\alpha(v_e) \isdef \el(e)$ for all $v_e \in V'$.
	According to~\cite{BBG-ColorRefinement}, one can compute a coarsest stable coloring
	$\beta'$ of $G'$ that refines $\alpha$ (and therefore also $\vl$) in time
	$O((|\VOne|+|V'|+|E'|) \cdot \log(|\VOne|+|V'|))$.
	In the case of $G'$ (a graph without edge-colors),
	stable means that for all $u,v$ such that $\beta'(u) = \beta'(v)$
	it holds that $| \NSucc{}{u}{c}| = |\NSucc{}{v}{c}|$ for all $c \in \img(\beta')$.
	Here, $\NSucc{}{u}{c}$ denotes the set of all outgoing
	$E'$-neighbors $w$ of $u$ such that $\beta'(w) = c$.
	\bigskip

	We will now show that a stable coloring of $\GOne$ corresponds to a stable coloring of $G'$
	and vice-versa.

	\medskip
	($\Rightarrow$):\;
	Let $\beta$ be a stable coloring of $\GOne$.
	Let $\beta'$ be the extension of $\beta$ to $G'$ where for all $v_{(u,w)} \in V'$
	we let $\beta'(v_{(u,w)}) \isdef (\el(u,w), \beta(w))$.
	Now we argue that $\beta'$ is stable on $G'$.
	Notice that the colors on $V$ and on $V'$ are distinct,
	therefore it suffices to consider the following two cases:
	
	\emph{Case 1}:\;
	Let $u,v \in \VOne$ such that $\beta'(u) = \beta'(v)$.
	Notice that $u$ and $v$ only have neighbors in $V'$,
	therefore for all colors $c \in \img(\beta)$ of the original coloring $\beta$,
	we have $|\NSucc{}{u}{c}| = |\NSucc{}{v}{c}| = 0$.
	Next, observe that for every color $c \in \img(\beta)$,
	there is a one-to-one correspondence between the set
	\begin{align*}
		\NSucc{\lambda}{u}{c} &=
			\set{ w \in \VOne \mid (u,w) \in \EOne, \ \el(u,w) = \lambda, \ \beta(w) = c }, \;
			\text{and the set} \\
		\NSucc{}{u}{d} &=
			\set{ v_{(u, w)} \mid (u,w) \in \EOne,\ \el(u,w) = \lambda,\ \beta(w) = c }, \;
			\text{where } d = (\lambda, c) \in \img(\beta').
	\end{align*}
	Thus, let $c \in \img(\beta')$ of the form $(\lambda, d)$.
	Then $|\NSucc{}{u}{c}| = |\NSucc{\lambda}{u}{d}| = |\NSucc{\lambda}{v}{d}| = |\NSucc{}{v}{c}|$.
	I.e., the coloring $\beta'$ is stable.

	\emph{Case 2}:\; Let $v_{e}, v_{f} \in V'$ such that $\beta'(v_{e}) = \beta'(v_{f})$.
	Since there is only one outgoing neighbor for both $v_{e}$ and $v_{f}$ by construction,
	it is easy to see that $|\NSucc{}{v_{e}}{d}| = 1$ and $|\NSucc{}{v_{f}}{d'}| = 1$
	for specific colors $d,d' \in \img(\beta)$ and 0 else.
	By construction, $\beta'(v_e) = (\lambda, d)$ and $\beta'(v_f) = (\lambda', d')$.
	Since $\beta'(v_e) = \beta'(v_f)$, $d = d'$ must hold.

	\medskip
	($\Leftarrow$):\; In the backwards-direction, let $\beta'$ be a stable coloring of $G'$
	and let $\beta$ be the restriction of $\beta'$ to $V$.
	We must argue, that $\beta$ is a stable coloring for $\GOne$.
	Notice that for all $v_{(u,v)}, v_{(u',v')} \in V'$ we have that
	$\beta'(v_{(u,v)}) = \beta'(v_{(u',v')})$ if and only if
	$\el(u,v) = \el(u',v')$ and $\beta'(v) = \beta'(v')$.

	Let $u,v \in \VOne$ such that $\beta(u) = \beta(v)$, let $c \in \img(\beta)$
	and let $\NSucc{\lambda}{u}{c} = \set{ w_1, \dots, w_{\ell} }$ for some $\lambda$.
	By the above observation, there must be a $d \in \img(\beta')$ such that
	$\NSucc{}{u}{d} = \set{ v_{(u,w_1)}, \dots, v_{(u,w_{\ell})} }$.
	Since $\beta'(u) = \beta'(v)$ this means that
	$\NSucc{}{v}{d} = \set{ v_{(v,w_1')}, \dots, v_{(v,w'_{\ell})} }$,
	which by the above observation also means that $\beta'(w'_i) = c$, i.e.,
	$\NSucc{\lambda}{v}{c} = \set{ w_1', \dots, w_{\ell}' }$.
	Hence, the coloring $\beta$ is stable.

	Finally, it is clear that if $\beta$ were not a coarsest stable coloring,
	then $\beta'$ would not be a coarsest stable coloring either.
	Thus, we can use the algorithm from~\cite{BBG-ColorRefinement} to compute
	a coarsest stable coloring $\beta$ for $\GOne$ that refines $\vl$ in time
	$O((|\VOne|+|V'|+|E'|) \cdot \log(|\VOne|+|V'|)) =
		O((|\VOne| + |\EOne|) \cdot \log(|\VOne| + |\EOne|)) =
		O((|\VOne| + |\EOne|) \cdot \log |\VOne|)$.
\end{proof}
 
\section{Missing Proofs in Section~\ref{sec:eval}}\label{sec:evalApp}
Let $D$, $Q \in \fcACQ$ be fixed in this section.
Recall that $\Vars(\QOne) = \Vars(\ciQ)$ and $G(\QOne) = G(\ciQ)$.

\subsection{Fundamental Observations for the Results of Section~\ref{sec:eval}}
\smallskip

\begin{observation}\label{obs:appendix:main-lemma:unary_predicates}
	For every color $c \in C$ and for every $u \in \Adom(\DOne)$ with $\col(u) = c$ we have: \ \
	$\set{ U \mid (c) \in U^{\ciD} } = \vl(u)$.
\end{observation}
\begin{proof}
	This trivially follows from the following facts:
	By definition, $(c) \in U^{\ciD}$ iff there exists a $v \in \Adom(\DOne)$ with
	$\col(v) = c$ and $(v)\in U^{\DOne}$.
	Furthermore, $(v)\in U^{\DOne}$ iff $U \in \vl(v)$.
	Finally, the coloring $\col$ refines $\vl$, i.e.,
	for all $u, v \in \Adom(\DOne)$ with $\col(u) = \col(v)$ we have $\vl(u) = \vl(v)$.
\end{proof}

For every variable $x \in \Vars(\QOne)$ other than the root $x_1$,
let $\Parent(x)$ be the unique \emph{parent} of $x$, that is the unique $z \in \Vars(\QOne)$
such that $z < x$ and $e \isdef \set{ x,z }$ is an edge in $G(\QOne)$.
We write $\ParentEdge(x)$ to denote the edge $\set{\Parent(x),x}$
connecting $x$ with its parent $\Parent(x)$.
Conversely, for all $x \in \Vars(\QOne)$, let $\Children(x)$ be the set of all children of $x$,
i.e., all $y \in \Vars(\QOne)$ such that $\Parent(y) = x$.

Using Lemma~\ref{lemma:homomorphisms}, one obtains:
\begin{corollary}\label{cor:homomorphisms}
	A mapping $\nu: \vars(\QOne) \to \Dom$ is a homomorphism from $\QOne$ to $\DOne$ if and only if:
	\begin{enumerate}[(1)]
		\item For all $x \in \Vars(\QOne)$ we have: \ $\lambda_x \subseteq \vl(\nu(x))$, and
		\item For all $x \in \Vars(\QOne)$ with $x \neq x_1$ we have:\ \,
		$\nu(x) \in \hatNSucc{\lambda}{\nu(y)}{d}$,
		where $y = \Parent(x)$, $\lambda = \lambda_{\ParentEdge(x)}$ and $d = \col(\nu(x))$.
	\end{enumerate}
\end{corollary}
\begin{proof}
	Condition (1) is the same as in Lemma~\ref{lemma:homomorphisms}.
	It remains to show that condition~(2) holds if and only if condition~(2) of
	Lemma~\ref{lemma:homomorphisms} holds.
	Notice that for every edge $e = \set{ x,y }$ of $G(\QOne)$ with $y < x$ it holds that
	$y = \Parent(x)$ and $\lambda_e = \lambda_{\ParentEdge(x)}$.
	Let $u = \nu(y)$ and $w = \nu(x)$.
	Then $(u,w) \in \EOne$ and $\el(u, w) \supseteq \lambda_e$ holds if and only if
	$w \in \NSucc{\lambda'}{u}{d}$ for $d = \col(w)$ and $\lambda' \supseteq \lambda_e$.
	This is the case if and only if $w \in \hatNSucc{\lambda_e}{u}{d}$, i.e.,
	iff $\nu(x) \in \hatNSucc{\lambda}{\nu(y)}{d}$.
\end{proof}

\begin{lemma}\label{claim:appendix:main-lemma:neighbors}
	Let $\mu: \Vars(\ciQ) \to C$ be a homomorphism from $\ciQ$ to $\ciD$.
	For every $x \in \Vars(\ciQ)$ that is not a leaf of $T$,
	for every $v \in \VOne$ with $\col(v) = \mu(x)$, and for every $z \in \Children(x)$ we have: \
	$\hatNSucc{\lambda}{v}{d} \neq \emptyset$, where $\lambda = \lambda_{\ParentEdge(z)}$
	and $d = \mu(z)$.
\end{lemma}
\begin{proof}
	Let $x \in \Vars(\ciQ)$ with $\Children(x) \neq \emptyset$,
	let $v \in \VOne$ such that $\col(v) = \mu(x) = f$ and let $z \in \Children(x)$.
	Let $\lambda = \lambda_{\ParentEdge(z)}$ and $d = \mu(z)$.

	Then, $E_\lambda(x, z) \in \Atoms(\ciQ)$ and since $\mu$ is a homomorphism,
	$(\mu(x), \mu(z)) \in E_{\lambda}^{\ciD}$, i.e., $(f, d) \in E_{\lambda}^{\ciD}$.
	This means that $\hatnumSucc{\lambda}{f}{d} > 0$,
	which means that for every $u \in \VOne$ with $\col(u) = f$,
	$\hatNSucc{\lambda}{u}{d} \neq \emptyset$,
	which means in particular that $\hatNSucc{\lambda}{v}{d} \neq \emptyset$.
\end{proof}

\begin{lemma}\label{prop:appendix:main-lemma:hom_correspondence}
	Let $\mu: \Vars(\ciQ) \to C$ and $\nu: \Vars(\QOne) \to \VOne$ be mappings such that
	for all $x \in \Vars(\QOne)$ we have that:
	\begin{enumerate}[(a)]
		\item\label{prop:appendix:main-lemma:hom_correspondence:a}
		$\col(\nu(x)) = \mu(x)$, and
		\item\label{prop:appendix:main-lemma:hom_correspondence:b}
		$x = x_1$ or $\nu(x) \in \hatNSucc{\lambda}{\nu(y)}{d}$
		where $y = \Parent(x)$, $\lambda = \lambda_{\ParentEdge(x)}$ and $d = \mu(x)$.
	\end{enumerate}
	Then, $\mu$ is a homomorphism from $\ciQ$ to $\ciD$ if and only if
	$\nu$ is a homomorphism from $\QOne$ to $\DOne$.
\end{lemma}

\begin{proof}
	\noindent($\Rightarrow$):\;
	We use Corollary~\ref{cor:homomorphisms} to verify that
	$\nu$ is a homomorphism from $\QOne$ to $\DOne$:

	\noindent (1):\; Let $x \in \Vars(\QOne)$ and let $\nu(x) = u$ and $\mu(x) = c$.
	We must show that $\vl(u) \supseteq \lambda_x$.
	Since $\mu$ is a homomorphism from $\ciQ$ to $\ciD$,
	$\lambda_x \subseteq \set{ U \mid (c) \in U^{\ciD} } =: A$.
	Since $\col(u) = c$ by~\eqref{prop:appendix:main-lemma:hom_correspondence:a},
	we have $A = \vl(u)$ using Observation~\ref{obs:appendix:main-lemma:unary_predicates}.
	Thus, $\lambda_x \subseteq \vl(u)$.
	\smallskip

	\noindent (2):\;
	This follows directly from~\eqref{prop:appendix:main-lemma:hom_correspondence:b}.

	\medskip
	\noindent($\Leftarrow$):\;
	Let $x \in \Vars(\ciQ)$ and let $\mu(x) = c$ and $\nu(x) = u$.
	Note that $\col(u) = c$.

	Suppose $U(x) \in \Atoms(\ciQ)$, then we must show that $(c) \in U^{\ciD}$.
	By definition of $\ciQ$ and $\QOne$, we also have $U(x) \in \Atoms(\QOne)$.
	It follows from Corollary~\ref{cor:homomorphisms} that $\vl(u) \supseteq \lambda_x$,
	and thus in particular $U \in \vl(u)$.
	According to Observation~\ref{obs:appendix:main-lemma:unary_predicates} it therefore also holds
	that $U \in \set{ U' \mid (c) \in U'^{\ciD} }$, i.e., $(c) \in U^{\ciD}$.

	Suppose there exists a $y \in \Vars(\ciQ)$ such that $E_{\lambda}(x,y) \in \Atoms(\ciQ)$.
	Then by definition of $\ciQ$, $x < y$ and $\lambda = \lambda_e$ for $e = \set{ x,y }$.
	Let $w \isdef \nu(y)$ and let $d \isdef \mu(y)$.
	Again, note that $\col(w) = d$.
	We must show that $(c, d) \in E_{\lambda}^{\ciD}$.
	It follows from Corollary~\ref{cor:homomorphisms} that
	$w \in \hatNSucc{\lambda}{u}{d}$, i.e., $\hatnumSucc{\lambda}{c}{d} > 0$.
	Thus, by definition of $\ciD$, $(c,d) \in E_{\lambda}^{\ciD}$.
\end{proof}

\subsection{Proof of Lemma~\ref{lemma:bool}}
If there is a homomorphism $\nu: \vars(\QOne) \to \VOne$ from $\QOne$ to $\DOne$,
then let $\mu: \vars(\ciQ) \to C$ with $\mu(x) = \col(\nu(x))$ for all $x \in \vars(\ciQ)$.
Then, according to Corollary~\ref{cor:homomorphisms}, the mappings $\mu$ and $\nu$ match
the requirements (a) and (b) of Lemma~\ref{prop:appendix:main-lemma:hom_correspondence},
which means $\mu$ is a homomorphism as well.

If there is a homomorphism $\mu: \vars(\ciQ) \to C$ from $\ciQ$ to $\ciD$,
we can combine Lemma~\ref{prop:appendix:main-lemma:hom_correspondence} with
Lemma~\ref{claim:appendix:main-lemma:neighbors}
to obtain a homomorphism $\nu: \Vars(\QOne) \to \VOne$ from $\QOne$ to $\DOne$.
\qed{}

\subsection{Preparing the proof of Lemma~\ref{lemma:mainLemma}}
For the next two lemmas, consider a fixed $\bar{c} \isdef (c_1, \dots, c_k) \in \sem{\ciQ}(\ciD)$
and a fixed homomorphism $\mu: \Vars(\ciQ) \to C$ from $\ciQ$ to $\ciD$,
that witnesses $\bar{c} \in \sem{\ciQ}(\ciD)$, i.e., $\mu(x_i) = c_i$ for all $i \in [k]$.
For $\ell \in [k]$, we call an $\ell$-tuple $(v_1, \dots, v_{\ell})$ \emph{consistent with}
$\bar{c} \isdef (c_1, \dots, c_k)$, if for all $i \in [\ell]$ $\col(v_i) = c_i$,
and all $j < i$ where $e \isdef \set{ x_i, x_j }$ is an edge in $G(\QOne)$
we have $v_i \in \hatNSucc{\lambda_e}{v_j}{c_i}$.
Note that $(v_1)$ is consistent with $\bar{c}$
for every $v_1 \in \Adom(\DOne)$ with $\col(v_1) = c_1$.

\begin{lemma}
	Let $\ell \in \set{ 1, \dots, k{-}1 }$
	and let $(v_1, \dots, v_{\ell})$ be consistent with $\bar{c}$.
	Let $x_j = \Parent(x_{\ell+1})$ and $\lambda = \lambda_{\ParentEdge(x_{\ell+1})}$.
	Then, $(v_1, \dots, v_{\ell}, v_{\ell+1})$ is consistent with $\bar{c}$
	for all $v_{\ell+1} \in \hatNSucc{\lambda}{v_j}{c_{\ell+1}}$.
\end{lemma}

\begin{proof}
	Let $v_{\ell+1} \in \hatNSucc{\lambda}{v_j}{c_{\ell+1}}$.
	It suffices to show that $(v_1, \dots, v_{\ell}, v_{\ell+1})$ is consistent.
	Let $i \in [\ell+1]$.
	Clearly, $\col(v_i) = c_{i}$ holds.
	Let $j < i$ such that $e \isdef \set{ x_i, x_j }$ is an edge in $G(\QOne)$.
	If $i \leq \ell$, then $v_i \in \hatNSucc{\lambda_e}{v_j}{c_i}$
	since $(v_1, \dots, v_{\ell})$ is consistent with $\bar{c}$.
	If $i = \ell+1$, then we have already seen that $x_j$ is unique.
	This means we have $x_j = \Parent(x_i)$ and $e = \ParentEdge(x_i)$.
	By choice of $v_i = v_{\ell+1}$, $v_{i} \in \hatNSucc{\lambda_e}{v_j}{c_i}$.

	Thus, $(v_1, \dots, v_{\ell}, v_{\ell+1})$ is consistent.
\end{proof}

\begin{lemma}\label{claim:appendix:main-lemma:consistency}
	For every $\ell \in [k]$ the following is true:
	Every $\ell$-tuple consistent with $\bar{c}$ is a partial output of $\QOne$ over $\DOne$
	of color $\bar{c}$; and vice versa,
	every partial output of $\QOne$ over $\DOne$ of color $\bar{c}$ is consistent with $\bar{c}$.
\end{lemma}
\begin{proof}
	Let $(v_1, \dots, v_{\ell})$ be consistent with $\bar{c}$.
	We show that there exists a homomorphism $\nu: \Vars(\QOne) \to \VOne$ from $\QOne$ to $\DOne$
	such that $\nu(x_i) = v_i$ for all $i \in [\ell]$ and $\col(\nu(z)) = \mu(z)$
	for all $z \in \Vars(\QOne)$.
	From that it directly follows that $(v_1, \dots, v_{\ell})$ is a partial output.

	We define $\nu$ along the order $<$ on $\Vars(\QOne)$ (which implies a top-down approach).
	For all $i \in [\ell]$, let $\nu(x_i) \isdef v_i$.
	
	For all $z \in \Vars(\QOne)$ such that $z \neq x_i$
	for all $i \in [\ell]$, let $y = \Parent(z)$.
	Since $y < z$, $\nu(y)$ is already defined and $\col(\nu(y)) = \mu(y)$.
	Let $u = \nu(y)$, $f = \mu(z)$ and $e \isdef \ParentEdge(z)$.
	By Lemma~\ref{claim:appendix:main-lemma:neighbors} we know that
	$\hatNSucc{\lambda_e}{u}{f}$ is not empty.
	We choose an arbitrary $w \in \hatNSucc{\lambda_e}{u}{f}$ and let $\nu(z) \isdef w$.
	Clearly, $\col(\nu(w)) = f = \mu(z)$.
	According to Lemma~\ref{prop:appendix:main-lemma:hom_correspondence}, $\nu$ is a homomorphism.
	\medskip
	
	Let $(v_1, \dots, v_{\ell})$ be a partial output of $\QOne$ over $\DOne$ of color $\bar{c}$.
	We have to show that it is consistent with $\bar{c}$.
	Let $i \in [\ell]$.
	Clearly, $\col(v_i) = c_i$.
	Let $j < i$ such that $e \isdef \set{ x_j, x_i }$ is an edge in $G(\QOne)$.
	Since $(v_1, \dots, v_{\ell})$ is a partial output, there exists a homomorphism
	$\nu$ from $\QOne$ to $\DOne$ such that $\nu(x_i) = v_i$ and $\nu(x_j) = v_j$.
	Notice that $x_j = \Parent(x_i)$, $e = \ParentEdge(x_i)$ and $\col(v_i) = c_i$.
	Thus, by Corollary~\ref{cor:homomorphisms} we get that $v_i \in \hatNSucc{\lambda_e}{v_j}{c_i}$.
\end{proof}

\subsection{Proof of Lemma~\ref{lemma:mainLemma}}
\begin{mea}
	\item
	Since $(v_1, \dots, v_k) \in \sem{\QOne}(\DOne)$, there exists a homomorphism
	$\nu: \Vars(\QOne) \to \VOne$ from $\QOne$ to $\DOne$ such that
	$\nu(x_i) = v_i$ for all $i \in [k]$.
	Let $\mu: \Vars(\ciQ) \to C$ with $\mu(x) = \col(\nu(x))$ for all $x \in \Vars(\ciQ)$.
	Using Corollary~\ref{cor:homomorphisms},
	we can apply Lemma~\ref{prop:appendix:main-lemma:hom_correspondence} to $\mu$ and $\nu$.
	Thus, $\mu$ is a homomorphism witnessing that $(\col(v_1), \dots, \col(v_k)) \in \sem{Q}(\ciD)$.
	\item
	For every $v_1 \in \Adom(\DOne)$ with $\col(v_1) = c_1$, $(v_1)$ is a partial output of
	$\QOne$ over $\DOne$ of color $\bar{c}$ because $(v_1)$ is consistent with $\bar{c}$.
	If $(v_1,\ldots,v_i)$ is a partial output of $\QOne$ over $\DOne$ of color $\bar{c}$,
	then it is also consistent with $\bar{c}$.
	If $e=\set{x_j,x_{i+1}}$ is an edge of $G(\QOne)$ with $x_j < x_{i+1}$,
	then $(v_1,\ldots,v_i,v_{i+1})$ is also consistent with $\bar{c}$ for every
	$v_{i+1} \in \hatNSucc{\lambda_e}{v_j}{c_{i+1}}$ and thus also a partial output of $\QOne$ over
	$\DOne$ of color $\ov{c}$ for every $v_{i+1} \in \hatNSucc{\lambda_e}{v_j}{c_{i+1}}$.
	$\hatNSucc{\lambda_e}{v_j}{c_{i+1}} \neq \emptyset$ holds by
	Lemma~\ref{claim:appendix:main-lemma:neighbors}.
	\qed{}
\end{mea}

\subsection{Proof of Lemma~\ref{lemma:counting}}
We prove this lemma by induction over the tree $T$.

\noindent\emph{Base case}:\;
Let $v \in \VOne$, $\col(v) = c$ and let $x$ be a leaf.
Then (b) trivially holds, since $\Children(x) = \emptyset$.
To prove (a), notice that $V(T_x) = \set{ x }$ and $E(T_x) = \emptyset$.
Thus, there is at most one $\nu: V(T_x) \to \VOne$ such that $\nu(x) = v$;
and it exists if and only if $\vl(\nu(x)) \supseteq \lambda_x$.
By definition, $\fd(c, x) = f_1(c, x)$ and $f_1(c,x) = 1$ if and only if
$\vl(v_c) \supseteq \lambda_x$ and 0 otherwise.
Since $\col(v) = c$, we have $\vl(v) = \vl(v_c)$ according to
Observation~\ref{obs:appendix:main-lemma:unary_predicates}.
Hence, condition (a)(1) holds.
Condition (a)(2) holds trivially, because $E(T_x) = \emptyset$.

\bigskip
\noindent\emph{Inductive step}:\;
Let $x \in T$ such that $x$ is not a leaf, and let $v \in \VOne$ with $\col(v) = c$.

\bigskip
\noindent\emph{Induction hypothesis}:\;
For every $y \in \Children(x)$, every $d \in C$ and every $w \in \VOne$ with $\col(w) = d$,
$\fd(d, y)$ is the number of mappings $\nu: V(T_y) \to \VOne$ satisfying $\nu(y) = w$ and
\begin{enumerate}[(1)]
	\item
	for every $x'\in V(T_y)$ we have $\vl(\val(x'))\supseteq \lambda_{x'}$, and
	\item
	for every edge $e=\set{x',y'}$ in $T_y$ with $x'<y'$ we have\\
	$(\val(x'),\val(y'))\in \EOne$ and $\el(\val(x'),\val(y'))\supseteq\lambda_e$.
\end{enumerate}
\bigskip

\noindent\emph{Claim}:\; For every $y \in \Children(x)$, $g(c, y)$ is the number of mappings
$\nu: \set{ x } \cup V(T_y) \to \VOne$ satisfying $\nu(x) = v$ and
\begin{enumerate}[(1)]
	\item
	for every $x'\in V(T_y)$ we have $\vl(\val(x'))\supseteq \lambda_{x'}$, and
	\item
	for every edge $e=\set{x',y'}$ in $T_y$ with $x'<y'$\\
	we have $(\val(x'),\val(y'))\in \EOne$ and $\el(\val(x'),\val(y'))\supseteq\lambda_e$, and
	\item
	for $e = \set{ x, y }$ we have $(\val(x),\val(y))\in \EOne$
	and $\el(\val(x),\val(y))\supseteq\lambda_e$.
\end{enumerate}
\smallskip

\begin{claimproof}
	Notice that (1) and (2) are the same conditions as in the induction hypothesis.
	Since every mapping has to map $x$ to $v$, every child $y \in \Children(x)$ has to be mapped to
	a $w$ such that $(v, w) \in \EOne$ and
	$\el(v, w) \supseteq \lambda_e$ for $e = \set{ x, y }$ according to (3), i.e.,
	we know that $y$ has to be mapped to an outgoing $\EOne$-neighbor $w$ of $v$
	such that $\el(v,w) \supseteq \lambda_e$.
	That is, any $w \in \hatNSucc{\lambda_e}{v}{c'}$ for any $c' \in C$ is a valid choice, i.e.,
	we have $|\hatNSucc{\lambda_e}{v}{c'}| = \hatnumSucc{\lambda_e}{c}{c'}$ choices for $w$
	for all $c' \in C$.
	Once we chose a $w$ in this way and let $\nu(y) = w$,
	we can use the induction hypothesis to get the number of choices we have to map
	the remaining variables such that $\nu(y) = w$ and (1) and (2) hold, which is $\fd(\col(w), y)$.
	Thus, we get that the number of mappings is
	$\sum_{c' \in C} \fd(c', y) \cdot \hatnumSucc{\lambda_e}{c}{c'}$,
	which is $g(c, y)$ by definition.
\end{claimproof}
\medskip

Notice that condition~(1) is the same as~(b)(1) and if we merge the conditions~(2) and~(3) above,
we get the condition~(2)(b).
Therefore, it follows from this claim that statement~(b) holds.
\medskip

\noindent\emph{Claim}:\;
$f_{\downarrow}(c,x)$ is the number of mappings $\val:V(T_x)\to \VOne$ satisfying $\val(x)=v$ and
\begin{enumerate}[(1)]
	\item
	for every $x'\in V(T_x)$ we have $\vl(\val(x'))\supseteq \lambda_{x'}$, and
	\item
	for every edge $e=\set{x',y'}$ in $T_x$ with $x'<y'$\\
	we have $(\val(x'),\val(y'))\in \EOne$ and $\el(\val(x'),\val(y'))\supseteq\lambda_e$.
\end{enumerate}
\smallskip
\begin{claimproof}
	In the same way as in the base case, we can see that the number of mappings is 0
	unless $f_1(c, x)$ is 1.
	If $f_1(c, x)$ is 1, then, according to~(2), a valid mapping must put every $y \in \Children(x)$
	on an outgoing $\EOne$-neighbor $w$ of $v$ such that $\el(v, w) \supseteq \lambda_e$
	where $e = \set{ x,y }$.

	Thus, every valid mapping must fulfill the requirements~(1)-(3) of the previous claim for every
	$y \in \Children(y)$ if we restrict it to $\set{ x } \cup T_{y}$.
	On the other hand, these trees only intersect in $x$ for all children.
	Thus, we can combine every map $\nu: \set{ x } \cup V(T_y) \to \VOne$ of a child $y$
	that adheres to these requirements with every other map $\nu': \set{ x } \cup V(T_{y'})$ of
	every other child $y'$ that adheres to these requirements,
	and we will get a map $\nu$ that satisfies (1) and (2).
	Thus, by the previous claim, we get a total of $\prod_{y \in \Children(x)} g(c,y)$ choices
	if $f_1(c,x)$ is 1 and 0 choices if $f_1(c,x)$ is 0, which is precisely $\fd(c, x)$.
\end{claimproof}
\medskip

It follows directly from this claim that statement (a) holds.\qed%

\section{A modified approach that avoids the exponential dependency on \texorpdfstring{$\sigma$}{sigma}}\label{appendix:MoreComplicatedIndexing}
Before finalizing this paper, we had two versions of our construction.
Version~1 is clean, easy to understand, and easily proved to be correct ---
but has an exponential dependency on the schema.
Version~2 is more complicated, more difficult to be proved to be correct ---
but avoids the exponential dependency on the schema.
Once having understood the construction of Version~1,
we believe that the overall idea of Version~2 can be explained with only moderate additional effort.
But presenting only Version~2 without prior explanation of Version~1, we believe,
the main idea of our index structure (and its correctness) would have been much harder to grasp.
Version~2 works as follows.

In the definition of the color-index structure $\DSD$,
instead of $\widehat{N}_{\rightarrow}^{\lambda}$ and $\widehat{\#}_{\rightarrow}^{\lambda}$,
we only need $N_{\rightarrow}^{\lambda}$ and $\#_{\rightarrow}^{\lambda}$
for each $\lambda \in \ELabels$.
But the definition of the color database $\ciD$ and the color schema $\cisigma$ are quite different
than in Section~\ref{sec:mainresult}: $\cisigma$ has only one binary relation symbol $E$,
all the unary relation symbols of $\sigmaOne$, and further unary relation symbols
$U_{(R,\rightarrow)}$ and $U_{(R,\leftarrow)}$ for each binary symbol $R$ in $\sigma$.

The color-database $\ciD$ is constructed as follows.
Initialize it to be the empty database.
For each unary symbol $S \in \sigmaOne$ and each color $c \in C$,
insert into the $S$-relation of $\ciD$ the element $c$
iff the $S$-relation of $\DOne$ contains a node of color $c$.
For each edge-label $\lambda \in \ELabels$ and all colors $c,c' \in C$ with
$\numSucc{\lambda}{c}{c'}>0$, insert into the active domain a new element $e_{\lambda,c,c'}$
and insert into the $E$-relation of $\ciD$ an edge from $c$ to $e_{\lambda,c,c'}$
and an edge from $e_{\lambda,c,c'}$ to $c'$.
Furthermore, for every $(R,\ast) \in \lambda$,
insert into the $U_{(R,\ast)}$-relation of $\ciD$ the element $e_{\lambda,c,c'}$.

In Section~\ref{sec:eval}, the definition of the color query $\ciQ$ has to be modified as follows.
First, initialize the body of $\ciQ$ to contain all
the unary atoms present in the body of $\QOne$.
Then, we enrich the body of the query as follows:
We loop through the edges $e$ of the Gaifman graph $G(\QOne)$ of $\QOne$.
For each $e$, we proceed as follows.
Introduce a new variable $z_e$.
Let $x<y$ be the two variables that are the endpoints of edge $e$.
Insert into the body of $\ciQ$ the new atoms $E(x,z_e), E(z_e,y)$ and, moreover,
$U_{(R,\ast)}(z_e)$ for every $(R,\ast) \in \lambda_e$
(where $\lambda_e$ is defined in the same way as in Section~\ref{sec:eval}).
This completes the definition of the body of $\ciQ$.
The head of $\ciQ$ contains the same variables as the head of $\QOne$,
and in addition to these, also the variables $z_e$ for all those edges $e$ of $G(\QOne)$
where both endpoints of $e$ belong to the head of $\QOne$.
For later use let us write $f$ for a bijection that maps all these edges $e$ to the numbers
$1, \dots ,k{-}1$, where $k$ is the number of variables in the head of $\QOne$.
The reason why we have to add these variables $z_e$ to the head of $\ciQ$ is to ensure that
the query $\ciQ$ is an fc-ACQ (this is needed for being able to use Theorem~\ref{thm:BGS-enum}).

Finally, the details on the evaluation phase for \enumtask{} have to be modified accordingly:
We use Theorem~\ref{thm:BGS-enum} to enumerate (with constant delay)
the result tuples of $\ciQ$ on $\ciD$.
Now, these tuples are of the form $t \isdef (c_1, \dots, c_k, e_1, \dots, e_{k{-}1})$,
where the $c_1, \dots, c_k$ are colors in $C$,
and the $e_1, \dots, e_{k{-}1}$ are elements of $\ciD$ of the form $e_{\lambda,c,c'}$.

In the enumeration algorithm presented in Section~\ref{sec:eval}
we now have to do the following modifications:
$\ov{c}$ is replaced by $t$.
And in the final for-loop we replace $\hat{N}_{\rightarrow}^{\lambda_e}$ by
$N_{\rightarrow}^{\mu_e}$, where $\mu_e$ is defined as follows.
Let $j \isdef f(e)$.
Let $\mu_e$ consist of exactly those $(R,\ast)$,
for which $e_j$ is in the $U_{(R,\ast)}$-relation of $\ciD$.
A careful inspection of the obtained algorithm shows that it indeed enumerates,
without duplicates, all the tuples in the result of $Q$ on $D$.

Note that also the proof of our counting result can be adapted in such a way
that it builds upon Version~2 instead of Version~1.
We plan to include details on Version~2 (as well as correctness proofs)
in the journal version of this paper.
  
\end{document}